\documentclass[11pt,reqno]{amsart}
\usepackage{amsfonts,amsmath,amssymb}
\newtheorem{thm}{Theorem}[section]
\newtheorem{proposition}[thm]{Proposition}
\newtheorem{corollary}[thm]{Corollary}
\newtheorem{lemma}[thm]{Lemma}
\newtheorem{definition}[thm]{Definition}

\newtheorem{remark}[thm]{Remark}

\usepackage{color}
\usepackage{graphicx}
\usepackage{epstopdf}

\title[Super Jack-Laurent polynomials]{Super Jack-Laurent  polynomials }
\author{ A.N. Sergeev}

\address{Department of Mathematics, Saratov State University, Astrakhanskaya 83, Saratov 410012  and National Research University Higher School of Economics, Russian Federation.}
 \email{SergeevAN@info.sgu.ru, asergeev@hse.ru}
\begin{document}
\maketitle
\begin{center} {Dedicated  to A.A. Kirillov \\ on the occasion of his 81th birthday }
\end{center}

\begin{abstract} Let $\mathcal{D}_{n,m}$ be the algebra of  quantum integrals  of the
deformed Calogero-Moser-Sutherland problem corresponding to the root
system of the Lie superalgebra $\frak{gl}(n,m)$. The  algebra $\mathcal{D}_{n,m}$ acts
naturally on the  quasi-invariant Laurent polynomials and we investigate
the corresponding spectral decomposition. Even for general value of the
parameter $k$  the spectral decomposition is not multiplicity free and we prove  that the image
of the algebra $\mathcal{D}_{n,m}$ in  the algebra of  endomorphisms  
of the generalised eigenspace is $k[\varepsilon]^{\otimes r}$  where $k[\varepsilon]$ 
is the algebra of dual numbers. The corresponding representation  is the  regular representation  of the algebra $k[\varepsilon]^{\otimes r}$.
\end{abstract}
\tableofcontents

\section{Introduction}  
It is  known that for  any root system  of a semi-simple  Lie algebra one can construct  the quantum (trigonometric) Calogero - Moser - Sutherland operator (see  for example \cite{OlsPer}). In particular if  the root system is  of the type $A_{n-1}$ then we have  the following operator 
\begin{equation}
\label{CM}
 {\mathcal L}^{(n)}_2=\sum_{i=1}^n
\left(x_{i}\frac{\partial}{\partial
x_{i}}\right)^2-k\sum_{ i < j}^n
\frac{x_{i}+x_{j}}{x_{i}-x_{j}}\left(
x_{i}\frac{\partial}{\partial x_{i}}-
x_{j}\frac{\partial}{\partial
x_{j}}\right)
\end{equation}
where $k$ is a complex parameter which we suppose is not rational number.
We are considering only trigonometric version of such an operator   since it has close connections with representation theory of the corresponding Lie algebra.

This operator is integrable in the sense  that there are enough differential operators  commuting with it. Let us denote by $\mathcal D_n$ the algebra consisting of all such operators (sometimes called integrals). A natural area of the action of the algebra $\mathcal D_{n}$  is the algebra of the symmetric Laurent polynomials $\Lambda^{\pm}_n=\Bbb C[x^{\pm1}_1,\dots,x^{\pm1}_n]^{S_n}$. There is the basis in this algebra consisting of joint eigenfunctions. These eigenfunctions are called Jack polynomials. The Jack polynomials has many applications in the combinatorics, representation theory  and mathematical physics.

In this paper  we investigate  the natural generalisation of the operator (\ref{CM}) to the case of  the   "deformed" root system $A_{n-1,m-1}$.  Instead of the symmetric group  it is  natural to consider  the group  $G$ generated by reflections with respect to a deformed bilinear form (see section $2$ for details). This group is well defined for all $k\in\Bbb C$ except  $k=-1$.  Actually the groupoid  from the paper \cite{Weyl} can be naturally described using the group $G$.

The corresponding operator  has the following form
$$ 
\mathcal L^{(n,m)}_2=\sum_{i=1}^n
\left(x_{i}\frac{\partial}{\partial
x_{i}}\right)^2+k\sum_{j=n+1}^{n+m}
\left(x_{j}\frac{\partial}{\partial
x_{j}}\right)^2-
k \sum_{i < j}^n
\frac{x_{i}+x_{j}}{x_{i}-x_{j}}\left(
x_{i}\frac{\partial}{\partial x_{i}}-
x_{j}\frac{\partial}{\partial x_{j}}\right)
$$
$$
+\sum_{n<i <
j}^{m+n} \frac{x_{i}+x_{j}}{x_{i}-x_{j}}\left(
x_{i}\frac{\partial}{\partial x_{i}}-
x_{j}\frac{\partial}{\partial x_{j}}\right)-
\sum_{i=1 }^n\sum_{
j=n+1}^{n+m} \frac{x_{i}+x_{j}}{x_{i}-x_{j}}\left(
x_{i}\frac{\partial}{\partial x_{i}}-k
x_{j}\frac{\partial}{\partial x_{j}}\right)
$$ 
where as before $k$ is not rational number. 

This operator has only partial symmetry with respect to  the group $G$, in other words the operator is only symmetric   with respect to  the subgroup $S_n\times S_m$  but it is  integrable and we will denote by $\mathcal D_{n,m}$ the corresponding algebra of integrals. A natural area of the action of the algebra $\mathcal D_{n,m}$ is the algebra of the  quasi-invariant Laurent polynomials
 $$
 \Lambda^{\pm}_{n,m}=\{f\in\Bbb C[x^{\pm1}_1,\dots,x^{\pm1}_{n+m}]^{S_n\times S_m}\mid s_{\alpha}f-f\in(\alpha^2),\, \alpha\in R_1\}
 $$ 
where 
$x_i=e^{\varepsilon_i},\ i=1,\dots, n+m
$
with a natural action of the group $G : ge^{\chi}=e^{g\chi}$  and $R_1$ is the set of {\it odd roots} (see for details section $2$). 

Let us point out the main difficulties in the deformed case.
The  first one  is that we need to use  quasi-invariants and also quasi-homomorphisms  in order to prove that the algebra  $\Lambda^{\pm}_{n,m}$ is preserved  by  the algebra $\mathcal D_{n,m}$. Actually we  interpret the Moser  matrix  as a linear operator on the vector space  of quasi-homomorphisms.
  
The second one is that it turns out  that even for general value of $k$ there is no  any eigenbasis in the algebra  $\Lambda^{\pm}_{n,m}$. Therefore  in this situation  we can only use the  decomposition into generalised eigenspaces. 

The third main difference with the classical case is that in  order to describe this  decomposition in terms of the weights and deformed root system we need to use some geometric language which actually  close to the language used by  S.  Kerov \cite{Kerov}  for Young diagrams.

In general there is no canonical form for more then one commuting operator.  
But in our case we are able to describe explicitly  the action of the algebra $\mathcal D_{n,m}$ 
on the generalised eigenspace. In order to prove the corresponding fact we need to use the infinite dimensional version of the trigonometric CMS operator. And it turns out that  in this situation  the  geometric language  of weights and geometric interpretation of Young diagrams are quite convenient. In particularly we give a geometric interpretation of the correspondence between bi-partitions from  $(n,m)$ cross (see Definition \ref{cross} section $5$) and the dominant weights of the Lie superalgebra $\frak{gl}(n,m)$ which was known before in the algebraic form \cite{ MVDJ,Moens}.
One of the main result of the paper can be formulated in the following way:

\begin{thm} \label{main}  Let $k$ be  not rational  number (in particular k assumed to be  non-zero). Then:

{\rm1)} $\Lambda^{\pm}_{n,m}$  as a module  over the algebra $\mathcal D_{n,m}$ can be decomposed into  direct sum of generalised eigenspaces
\begin{equation}
\label{sum}
\Lambda^{\pm}_{n,m}=\bigoplus_{\chi\in X_{reg}^+(n,m)}\Lambda^{\pm}_{n,m}(\chi),
\end{equation}
where $X_{reg}^+(n,m)$ is defined in section $4$ definition \ref{reg}.

$2)$ The dimension  of the space $\Lambda^{\pm}_{n,m}(\chi)$ is equal to $2^r$ where $r$ is the number of odd positive roots $\alpha$ such that $(\chi+\rho,\alpha)+\frac12(\alpha,\alpha)=0$
and  $\rho$ is defined in section $3$ definition \ref{rho}.

$3)$ The algebra $\Lambda^{\pm}_{n,m}$ is generated by the deformed power sums
$$
p_s(x_1,\dots,x_{n+m})=x^s_1+\dots+x_n^s+\frac1k(x^s_{n+1}+\dots+x_{m+n}^s),\,\,s=\pm1,\pm2,\dots
$$

$4)$  If $k$ is not algebraic number  then  the image of the algebra  $\mathcal D_{n,m}$ 
in the algebra    $End(\Lambda^{\pm}_{n,m}(\chi))$ is isomorphic to 
$\Bbb C[\varepsilon]^{\otimes r}$, where 
$\Bbb C[\varepsilon]$ with $\varepsilon^2=0$ is the algebra of dual numbers  
and $\Lambda^{\pm}_{n,m}(\chi)$ is the regular representation relative to this action.
\end{thm}

The paper is organised  in the following way. In the section $2$  we define the group $G$ and  the notion of the quasi-homomorphism. Then we  prove the main result of this section which states that the Moser matrix is actually a linear operator on the space of quasi-homomorphisms.

In the section $3$ we reproduced  one of the results of the paper \cite{Symmetric} 
about spectral decomposition. We give a  more  conceptual  and technically  
more clear proof of some of the key steps. We also show that  the image of the algebra of integrals under the Harish-Chandra homomorphism is the algebra of quasi-invariants with respect  to  the group $G$.

In section $4$ we describe the equivalence classes in the spectral decomposition using some geometric interpretation of the dominant weights and then we translate these results into language  of deformed root systems.

In section $5$ we  give a geometric interpretation  of the correspondence between dominant  weights and bi-partitions which was known before in the algebraic form.

 In section $6$   we explicitly describe the action of the algebra of integrals in the generalised eigenspace, 
using   some of the main results about  the infinite dimensional  CMS operator.  In the paper \cite{Laurent} it was introduced and studied a Laurent version of Jack symmetric functions - {\it Jack--Laurent symmetric functions} $P_{\lambda,\mu}$ as certain elements of $\Lambda^{\pm}$ depending on complex parameters $k,p_0$ and   labelled by  pairs of the  partitions $\lambda$ and $\mu.$ Here $\Lambda^{\pm}$ is freely generated by $p_i$ with $i \in \mathbb Z \setminus \{0\}$ being both positive and negative.
The usual Jack symmetric functions $P_{\lambda}$  are particular cases of $P_{\lambda,\mu}$ corresponding to empty second partition $\mu.$ In the paper \cite{Laurent} it was  proved the existence of $P_{\lambda,\mu}$ for all $k \notin \mathbb Q$ and $p_0\neq n+k^{-1}m, \, m,n \in \mathbb Z_{>0}$. The  special case $p_0= n+k^{-1}m, \, m,n \in \mathbb Z_{>0}$ was studied in the paper \cite{Special}. 

\section{Quasi-invariants and quasi-homomorphisms} 
Let $I=\{1,\dots, n+m\}$ be  a set of indices with the parity  function $p(i)=0$ if $1\le i\le n$ and $p(i)=1$ if $n<i\le n+m$ and we suppose that  $0,1\in\Bbb Z_2$.
Let  also $V$ be a  vector space of dimension $n+m$   with a  basis $\varepsilon_1,\dots,\varepsilon_{n+m}$ and  a bilinear symmetric form 
$
(\varepsilon_i,\varepsilon_j)=\delta_{ij}k^{p(i)}
.$
We always suppose that $k$ is not rational.
Let us also denote by $R,\,R^+$  the set of {\it roots}  and {\it positive roots}
$$
R=\{\varepsilon_i-\varepsilon_j\,|\,i,j\in I; \,i\ne j\},
\quad 
R^+=\{\varepsilon_i-\varepsilon_j\,|\,i,j\in I; \,i<j\}
$$
The set $R$ can be naturally represented as a disjoint union  $R=R_0\cup R_1$
of the even and odd roots, where the parity of the root $\varepsilon_i-\varepsilon_j$ is $p(i)+p(j)$. 
\begin{definition} Let $G$ be the group generated by reflections $s_{\alpha},\alpha\in R$ 
where
$$
s_{\alpha}(v)=v-\frac{2(v,\alpha)}{(\alpha,\alpha)}\alpha
$$
\end{definition}
\begin{remark} Let us note, that if $k\ne-1$ then $(\alpha,\alpha)=1+k\ne0$ for any odd root $\alpha$  and the definition makes sense. Besides if  $k=1$ then $G=S_{n+m}$ is the symmetric group. 
We can also construct a groupoid   which set of objects is $R_1$ 
and the set of morphism from $\alpha$ to $\beta$ is $w\in S_n\times S_n\subset G$ 
such that $w(\alpha)=\beta$ or $s_{\beta}w$ where $w(\alpha)=-\beta$. 
It is easy to see that if we add to this groupoid   group $S_n\times S_m$ 
as one point groupoid then we get the same groupoid as in \cite{Weyl}. 
This  groupoid is important when $k=-1$, because the group $G$ 
is not well defined then, but groupoid still makes sense. 
\end{remark}
\begin{definition}\label{exp}
We will denote by $\Bbb C[[V]]$ the algebra of formal power series in 
$\varepsilon_1,\dots,\varepsilon_{n+m}$ and by $\hat S(V)$ its sub-algebra 
spanned by $e^{v}$ where $v\in V$ and the vector $v$ has integer 
coordinates with respect to the basis $\varepsilon_1,\dots,\varepsilon_{n+m}$. 
\end{definition}

Since the  group $G$ acts on $V$ it therefore naturally acts on $\Bbb C[[V]]$. 
In particular $g(e^v)=e^{gv}$ for $g\in G$. 
It is easy to see that the subgroup generated 
by $s_{\alpha},\alpha\in R_0$  is the product  $S_n\times S_m$ 
of two symmetric groups.

\begin{definition} An element $f\in \Bbb C[[V]]$ 
is called quasi-invariant for the group $G$ if it is invariant with respect to  $S_n\times S_m$ 
and  for each $\alpha\in R_1$ we have $s_{\alpha}f-f\in(\alpha^2)$ 
where $(\alpha^2)$ means the ideal generated by $\alpha^2$. 
\end{definition} 
We also  need a more general definition.
\begin{definition}
A linear map $\varphi: V\longrightarrow \Bbb C[[V]]$ is 
called quasi-homomorphism if it commutes with the action of  the subgroup $S_n\times S_m$ and  for $\alpha\in R_1$ and $v\in V$ we have 
$$
s_{\alpha}\varphi(v)-\varphi(s_{\alpha}v)\in (\alpha^2)
$$
\end{definition}
We should mention that the notion of quasi-invaraint was introduced by  O.Chalykh and A. Veselov (see \cite{quasi1}) and this notion   is related to  the notion of quasi-homomorphism. Namely, if  $s_{\alpha}v=v$ and $\varphi$ is a quasi-homomorphism, then $\varphi(v)$ is a quasi-invariant with respect to $s_{\alpha}$.
The above definitions can be given in infinitesimal form and they will also work for $k=-1$.
\begin{lemma}\label{tr}
$1)$  An element $f\in \Bbb C[[V]]$ is  a quasi-invariant if and only if $\partial_{\alpha}f\in(\alpha)$
where $\partial_{\alpha}(v)=(\alpha,v)$.

$2)$ A linear map 
$\varphi: V\longrightarrow \Bbb C[[V]]$ is a quasi-homomorphism if and only if 
\begin{equation}\label{infquasi}
\partial_{\alpha}\varphi(v)-\frac{(\alpha,v)\varphi(\alpha)}{\alpha}\in(\alpha)
\end{equation}

$3)$ An element $f\in \hat S(V)$ is  a quasi-invariant if and only if $\partial_{\alpha}f\in(e^{\alpha}-1)$.

$4)$  A linear map 
$\varphi: V\longrightarrow \hat S(V)$ is a quasi-homomorphism if and only if 
\begin{equation}\label{trquasi}
\partial_{\alpha}\varphi(v)-\frac{(\alpha,v)\varphi(\alpha)}{e^{\alpha}-1}\in(e^{\alpha}-1)
\end{equation}
\end{lemma}
\begin{proof} The first statement is a particular case of the second one. 
So, it is enough to prove the latter statement. It is easy to see  that for any  $f\in \Bbb C[[V]]$ we have the following equality
$$
s_{\alpha}f=f-\frac{2}{(\alpha,\alpha)}(\partial_{\alpha}f)\alpha+f_1\alpha^2,\,f_1\in \Bbb C[[V]]
$$
Therefore 
$$
s_{\alpha}\varphi(v)-\varphi(s_{\alpha} v)=\varphi(v)-\frac{2}{(\alpha,\alpha)}\partial_{\alpha}\varphi(v)\alpha-\varphi(v)+\frac{2(\alpha,v)}{(\alpha,\alpha)}\varphi(\alpha)+f_1\alpha^2
$$
$$
=-\frac{2}{(\alpha,\alpha)}\left[ \partial_{\alpha}\varphi(v)\alpha-(\alpha,v)\varphi(\alpha)\right]+f_1\alpha^2
$$
Now, let us prove the third statement. 
It is clear that the condition $\partial_{\alpha}f\in(e^{\alpha}-1)$ 
implies the condition $\partial_{\alpha}f\in(\alpha)$. 
Let us prove the converse statement. We can assume that $n=m=1$. Let $f\in\hat S(V)$ and $\partial_{\alpha}f\in (\alpha)$. Consider homomorphism
$$
\psi : \Bbb C[[\varepsilon_1,\varepsilon_2]]\longrightarrow \Bbb C[[t]],\,\psi(\varepsilon_1)=\psi(\varepsilon_2)=t
$$
We can write 
$$
f=\sum_{\lambda_1,\lambda_2\in \Bbb Z}c_{\lambda_1,\lambda_2}e^{\lambda_1\varepsilon_1+\lambda_2\varepsilon_2}
$$
From the condition $\partial_{\alpha}f\in (\alpha)$ it follows that 
$$
\psi(f)=\sum_{\lambda_1,\lambda_2\in \Bbb Z}c_{\lambda_1,\lambda_2}e^{(\lambda_1+\lambda_2)t}=0
$$
Therefore 
$$
f=\sum_{\lambda_1,\lambda_2\in \Bbb Z}c_{\lambda_1,\lambda_2}e^{\lambda_1\varepsilon_1+\lambda_2\varepsilon_2}=\sum_{\lambda_1,\lambda_2\in \Bbb Z}c_{\lambda_1,\lambda_2}\left(e^{\lambda_1\varepsilon_1+\lambda_2\varepsilon_2}-e^{(\lambda_1+\lambda_2)\varepsilon_1}\right)
$$
$$
=\sum_{\lambda_1,\lambda_2\in \Bbb Z}c_{\lambda_1,\lambda_2}e^{\lambda_1\varepsilon_1+\lambda_2\varepsilon_2}\left(1-e^{\lambda_2(\varepsilon_1-\varepsilon_2)}\right)\in\left(e^{\alpha}-1\right)
$$
and the third statement follows.

Now, let us prove the fourth statement.  It is easy to see  
that the formula (\ref{infquasi}) is equivalent  to the following two conditions: 

a) if $(\alpha,v)=0$, then  $\partial_{\alpha}\varphi(v)\in (\alpha)$;\ 

b) $\varphi(\alpha)\in(\alpha)$.
\\
Indeed, it is easy to see that the conditions a) and b) follow 
from the formula (\ref{infquasi})  if $k\ne-1$.  Let us prove now that the conditions a) and b) imply the formula (\ref{infquasi}).
Since the formula (\ref{infquasi}) is linear with respect to $v$ we need to check it in the cases a) and b). In the case $(\alpha,v)=0$ the condition  a)   is equivalent to the formula (\ref{infquasi}). In the case b) formula (\ref{infquasi}) has the form
$$
\partial_{\alpha}\varphi(\alpha)-\frac{(\alpha,\alpha)\varphi(\alpha)}{\alpha}=\partial_{\alpha}\left(\frac{\varphi(\alpha)}{\alpha}\right)\alpha.
$$
Therefore if $v=\alpha$  the formula (\ref{infquasi}) follows from the condition b). 
Now  using  the same arguments as in the proof of the  third statement   it is not difficult to deduce 
that in the case $\varphi(V)\subset \hat S(V)$ these two conditions are equivalent to 
 to another two:

a${}^*$) if $(\alpha,v)=0$, then  $\partial_{\alpha}\varphi(v)\in (e^{\alpha}-1)$

b${}^*$) $\varphi(\alpha)\in (e^{\alpha}-1)$.
\\
It is also easy to check that 
the  conditions a${}^*$) and b${}^*$) are equivalent to the formula (\ref{trquasi}) in the Lemma.
\end{proof}
Let us denote by $x_i=e^{\varepsilon_i},\,i=1,\dots, n+m$.
Then the algebra $\hat S(V)$ can be identified with the algebra of  Laurent polynomials $\Bbb C[x^{\pm1}_1,\dots,x^{\pm1}_{n+m}]$. Now 
introduce the following {\it  quantum Moser matrix} as   $(n+m)\times (n+m)$-matrix   $L$ by the formulae
\begin{equation}\label{Moser}
L_{ii}=\partial_{\varepsilon_i}-\sum_{j\ne i}k^{1-p(j)}\frac{x_{i}}{x_{i}-x_{j}},\,\,\, L_{ij}=k^{1-p(j)}\frac{x_i}{x_i-x_j},\, i\ne j
\end{equation}
If $\varphi: V\rightarrow \hat S(V)$ is a linear map, then we can define the map $\psi$ by the following fromula
\begin{equation}\label{quasi}
\begin{pmatrix}\psi(\varepsilon_1)\\ \vdots\\ \psi(\varepsilon_{m+n})\end{pmatrix}=\begin{pmatrix}L_{1,1}&\dots& L_{1,n+m}\\
L_{2,1}&\dots &L_{2,n+m}\\
\dots&\dots&\dots\\
L_{n+m,1}&\dots&
 L_{n+m,n+m} \end{pmatrix}\begin{pmatrix}\varphi(\varepsilon_1)\\
 \vdots\\
 \varphi(\varepsilon_{n+m})\end{pmatrix}
\end{equation}
\begin{thm}\label{quasi1} 
Let  $\varphi: V\rightarrow \hat S(V)$ be a quasi-homomorphism.
Then $\psi$ has a unique 
linear extension to the whole space $V$,
and this extension is a quasi-homomorphism to $\hat S(V)$.
\end{thm}
 \begin{proof} We will prove the theorem in several steps. 
First we will prove it in the case $n=m=1$. 
In this case we have 
$$
L=
\begin{pmatrix}\partial_{\varepsilon_1}-\frac{x_1}{x_1-x_2}&\frac{x_1}{x_1-x_2}\\
\frac{kx_2}{x_2-x_1}&k\partial_{\varepsilon_2}-\frac{kx_2}{x_2-x_1}
\end{pmatrix}
$$
We also have the following equalities
$$
\partial_{\varepsilon_1}=\frac{k}{1+k}\partial_{v^+}+\frac{1}{1+k}\partial_{\alpha},\,\,\,\,\partial_{\varepsilon_2}=\frac{k}{1+k}\partial_{v^+}-\frac{k}{1+k}\partial_{\alpha}
$$
where $v^+=\varepsilon_1+\frac1k\varepsilon_2$. Therefore we have the following equalities
$$
\begin{cases}\psi(\varepsilon_1)=\frac{1}{1+k}\partial_{\alpha}\varphi(\varepsilon_1)-\frac{x_1}{x_1-x_2}\varphi(\alpha)+\frac{k}{1+k}\partial_{v^+}\varphi(\varepsilon_1)\\
\psi(\varepsilon_2)=-\frac{k}{1+k}\partial_{\alpha}\varphi(\varepsilon_2)-\frac{kx_2}{x_1-x_2}\varphi(\alpha)+\frac{k}{1+k}\partial_{v^+}\varphi(\varepsilon_2)
\end{cases}
$$
It is easy to check that the correspondence  
$$
\varepsilon_1\longrightarrow \partial_{v^+}\varphi(\varepsilon_1),\,\,\, \varepsilon_2\longrightarrow \partial_{v^+}\varphi(\varepsilon_2)
$$
is a quasi-homomorphism. So we need to prove  that the formulae 
$$
\begin{cases}\psi(\varepsilon_1)=\frac{1}{1+k}\partial_{\alpha}\varphi(\varepsilon_1)-\frac{x_1}{x_1-x_2}\varphi(\alpha)\\
\psi(\varepsilon_2)=-\frac{k}{1+k}\partial_{\alpha}\varphi(\varepsilon_2)-\frac{kx_2}{x_1-x_2}\varphi(\alpha)
\end{cases}
$$
give a quasi-homomorphism. Let us check first that $\psi(\alpha)\in (x_1-x_2)$. We have 
$$
\psi(\alpha)=\frac{1}{1+k}\left(\partial_{\alpha}\varphi(\varepsilon_1)-\frac{x_2}{x_1-x_2}\varphi(\alpha)\right)+\frac{k}{1+k}\left(\partial_{\alpha}\varphi(\varepsilon_2)+\frac{kx_2}{x_1-x_2}\varphi(\alpha)\right)
$$
$$
+\frac{1}{1+k}\frac{x_2}{x_1-x_2}\varphi(\alpha)-\frac{k^2}{1+k}\frac{x_2}{x_1-x_2}\varphi(\alpha)-\frac{x_1}{x_1-x_2}\varphi(\alpha)+\frac{kx_2}{x_1-x_2}\varphi(\alpha)
$$
 Since $\varphi$ is a quasi-homomorphism, then two summands in big brackets are in the ideal  $(x_1-x_2)$. And the last expression can be simplified to the form $-\varphi(\alpha)$ and therefore belongs to the ideal 
$(x_1-x_2)$. So $\psi(\alpha)\in (x_1-x_2)$. And we only need to prove that $\partial_{\alpha}\psi(v^+)\in(x_1-x_2)$.
 We have 
$$
\psi(v^+)=\psi(\varepsilon_1)+\frac1k\psi(\varepsilon_2)=\frac{1}{1+k}\partial_{\alpha}\varphi(\varepsilon_1)-\frac{x_1}{x_1-x_2}\varphi(\alpha)
$$
$$
+\frac1k\left( -\frac{k}{1+k}\partial_{\alpha}\varphi(\varepsilon_2)-\frac{kx_2}{x_1-x_2}\varphi(\alpha)\right)=\frac{1}{1+k}\partial_{\alpha}\varphi(\alpha)-\frac{x_1+x_2}{x_1-x_2}\varphi(\alpha)
$$
We know, that $\varphi(\alpha)=f(x_1-x_2)$. Therefore the previous formula can be rewritten in the form
$$
\psi(v^+)=\frac{1}{1+k}\left[ \partial_{\alpha}f(x_1-x_2)-(kx_1+x_2)f\right]
$$
And it is easy to verify that
$$
\partial_{\alpha}\psi(v^+)=\frac{1}{1+k}\left[\partial_{\alpha}^2f+(1-k)\partial_{\alpha}f\right](x_1-x_2)\in(x_1-x_2)
$$
So the case $n=m=1$ is completely proved.

Now let us proceed to the general case. 
It is easy to check that matrix elements of $L$ 
satisfy the relations
 $L_{\sigma(i),\sigma(j)}=L_{ij}$ for $\sigma\in S_n\times S_m$. Therefore  by Lemma  \ref{tr} we only need to prove that if $\alpha\in R_1$, then $\psi(\alpha)\in (e^{\alpha}-1)$ and $\partial_{\alpha}(\psi(v))\in(e^{\alpha}-1)$ for $(\alpha,v)=0$. Let us take $\alpha=\varepsilon_i-\varepsilon_j$, where $1\le i\le n,\,\, n+1\le j\le n+m$ and prove that $\psi(\alpha)\in (x_i-x_j)$. We have 
$$
\psi(\varepsilon_i)=\partial_{\varepsilon_i}\varphi(\varepsilon_i)-\sum_{r\ne i}\frac{k^{1-p(r)}x_i}{x_i-x_r}(\varphi(\varepsilon_i)-\varphi(\varepsilon_r))$$
$$
=\partial_{\varepsilon_i}\varphi(\varepsilon_i)-\frac{x_i}{x_i-x_j}\varphi(\alpha)
-\sum_{r\ne i,j}\frac{k^{1-p(r)}x_i}{x_i-x_r}(\varphi(\varepsilon_i)-\varphi(\varepsilon_r))
$$
In the same way we have 
$$
\psi(\varepsilon_j)=\partial_{\varepsilon_j}\varphi(\varepsilon_j)-\frac{kx_j}{x_i-x_j}\varphi(\alpha)
-\sum_{r\ne i,j}\frac{k^{1-p(r)}x_j}{x_i-x_r}(\varphi(\varepsilon_j)-\varphi(\varepsilon_r))
$$
Therefore
$$
\psi(\alpha)=\partial_{\varepsilon_i}\varphi(\varepsilon_i)-\partial_{\varepsilon_j}\varphi(\varepsilon_j)-\frac{x_i-kx_j}{x_i-x_j}\varphi(\alpha)
$$
$$
+\sum_{r\ne i,j}\frac{k^{1-p(r)}x_i}{x_i-x_r}\varphi(\varepsilon_i-\varepsilon_r)-\sum_{r\ne i,j}\frac{k^{1-p(r)}x_j}{x_j-x_r}\varphi(\varepsilon_j-\varepsilon_r)
$$
And it is easy to check the following identity
$$
\frac{x_i\varphi(\varepsilon_i-\varepsilon_r)}{x_i-x_r}-\frac{x_j\varphi(\varepsilon_j-\varepsilon_r)
}{x_j-x_r}=\frac{x_i\varphi(\alpha)}{x_i-x_r}-\frac{x_r(x_i-x_j)\varphi(\varepsilon_j-\varepsilon_r)}
{(x_i-x_r)(x_j-x_r)}
$$
So we have proved  the condition b${}^*$) for quasi-homomorphism. 
Let us prove the condition a${}^*$). We need to consider two different cases:  first $v=\varepsilon_i+\frac1k\varepsilon_j $ and the second one $v=\varepsilon_s,\, s\ne i,j$.

In the first case we have 
$$
\psi(\varepsilon_i+\frac1k\varepsilon_j)=\partial_{\varepsilon_i}\varphi(\varepsilon_i)+\partial_{\varepsilon_j}\varphi(\varepsilon_j)-\frac{x_i+x_j}{x_i-x_j}\varphi(\alpha)
$$
$$
-\sum_{r\ne i,j}\frac{k^{1-p(r)}x_i}{x_i-x_r}(\varphi(\varepsilon_i)-\varphi(\varepsilon_r))+\frac{k^{-p(r)}x_j}{x_j-x_r}(\varphi(\varepsilon_j)-\varphi(\varepsilon_r))
$$
Since the case $n=m=1$ has been already considered, we only need to prove  that 
$$
\partial_{\alpha}\left(\frac{kx_i}{x_i-x_r}(\varphi(\varepsilon_i)-\varphi(\varepsilon_r))+\frac{x_j}{x_j-x_r}(\varphi(\varepsilon_j)-\varphi(\varepsilon_r))\right)\in (x_i-x_j)
$$
We have 
$$
\frac{kx_i}{x_i-x_r}(\varphi(\varepsilon_i)-\varphi(\varepsilon_r))+\frac{x_j}{x_j-x_r}(\varphi(\varepsilon_j)-\varphi(\varepsilon_r))
$$
$$
=\left(\frac{kx_i}{x_i-x_r}+\frac{x_j}{x_j-x_r}\right)\varphi(\varepsilon_r)-\left(\frac{kx_i}{x_i-x_r}\varphi(\varepsilon_i)+\frac{x_j}{x_j-x_r}\varphi(\varepsilon_j)\right)
$$
Further we  have 
$$
\partial_{\alpha}\left[\left(\frac{kx_i}{x_i-x_r}+\frac{x_j}{x_j-x_r}\right)\varphi(\varepsilon_r)\right]=\left(\frac{kx_i}{x_i-x_r}+\frac{x_j}{x_j-x_r}\right)\partial_{\alpha}(\varphi(\varepsilon_r))
$$
$$
+k\left(\frac{x_i}{x_i-x_r}-\frac{x_i^2}{(x_i-x_r)^2}-\frac{x_j}{x_j-x_r}+\frac{x_j^2}{(x_i-x_r)^2}\right)\in(x_i-x_j)
$$
since $\partial_{\alpha}\varphi(\varepsilon_r)\in(x_i-x_r)$.

In the second case we have  
$$
\psi(\varepsilon_s)=\partial_{\varepsilon_s}\varphi(\varepsilon_s)-\sum_{r\ne s,i,j}\frac{k^{1-p(r)}x_s}{x_s-x_r}(\varphi(\varepsilon_s)-\varphi(\varepsilon_r))
$$
$$
-\frac{kx_s}{x_s-x_i}(\varphi(\varepsilon_s)-\varphi(\varepsilon_i))-\frac{x_s}{x_s-x_j}(\varphi(\varepsilon_s)-\varphi(\varepsilon_j))
$$
It is easy to see that $\partial_{\alpha}(\partial_s\varphi(\varepsilon_s))\in(x_i-x_j)$. So we only need to prove the same for  the last two summands. But we have 
$$
\frac{kx_s}{x_s-x_i}(\varphi(\varepsilon_s)-\varphi(\varepsilon_i))+\frac{x_s}{x_s-x_j}(\varphi(\varepsilon_s)-\varphi(\varepsilon_j))
$$
$$
=\left[\frac{k}{x_s-x_i}+\frac{1}{x_s-x_j}\right]x_s\varphi(\varepsilon_s)-\left[\frac{k}{x_s-x_i}\varphi(\varepsilon_i)+\frac{1}{x_s-x_j}\varphi(\varepsilon_j)\right]x_s
$$
and
$$
\partial_{\alpha}\left[\frac{k}{x_s-x_i}+\frac{1}{x_s-x_j}\right]=\left[\frac{kx_i}{x_s-x_i}-\frac{kx_j}{x_s-x_j}\right]\in(x_i-x_j)
$$
and
$$
\partial_{\alpha}\left[\frac{k}{x_s-x_i}\varphi(\varepsilon_i)+\frac{1}{x_s-x_j}\varphi(\varepsilon_j)\right]=\frac{kx_s(x_i-x_j)}{(x_s-x_i)(x_s-x_j)}\varphi(\varepsilon_i)+\frac{kx_j}{x_s-x_j}\varphi(\alpha)
$$
The case $s>m$ can be considered in the same manner. 
\end{proof}
\section{Algebra of deformed CMS integrals and spectral decomposition}
Let us define  CMS integrals ${\mathcal{L}}_{r},\,r=1,2,\dots$ by the following formulae 
\begin{equation}
\label{dif2} {\mathcal{L}}_{r}=e^*L^re.
\end{equation}
where $e^*$  is a row  such, that $e^*_i=k^{-p(i)},\,i=1.\dots,m+n$ and  $e$ is a column such that  $e_i=1,\, i=1,\dots, n+m$ and $L$ is the Moser matrix given by (\ref{Moser}).
\begin{thm}
 The operators ${\mathcal{L}}_{r}$ are quantum integrals of the deformed CMS system:
$$
[ {\mathcal{L}}_{r},  {\mathcal{L}}_{2}]=0.
$$
\end{thm}
For the proof of the Theorem see \cite {Symmetric} Theorem 2.1.
\begin{definition} Let us denote by $\mathcal D_{n,m}$ the algebra generated  by operators ${\mathcal{L}}_{r},\,r=1,2,\dots$.
\end{definition}
Define now the {\it Harish-Chandra homomorphism}
$$
\phi_{n,m}: \mathcal D_{n,m} \rightarrow S(V^*)
$$
by the conditions:
$$
\phi_{n,m} (\partial_{\varepsilon_i})=k^{p(i)}\varepsilon^*_i, \quad \phi_{n,m} \left(\frac{x_{i}}{x_{i}-x_{j}}\right)=1, \, \,\, {\text {if}} \,\, i<j. 
$$
where $\varepsilon_i^*,\,i=1,\dots, n+m$ is the basis dual to  the basis $\varepsilon_i,\,i=1,\dots, n+m$.
\begin{definition}\label{rho}Let $\rho \in V$ be the following deformed analogue of the Weyl vector 
\begin{equation}
\label{rhok}
\rho=\frac12\sum_{i=1}^n(k(2i-n-1)-m)\varepsilon_i+\frac12\sum_{j=1}^{m}(k^{-1}(2j-m-1)+n)\varepsilon_{j+n}
\end{equation}
\end{definition}
Let us define an affine action of the group $G$ on the space $V.$ Namely 
\begin{equation}\label{action}
g\circ v=g(v+\rho)-\rho
\end{equation}
It is easy to see that under this affine action the vector $-\rho$ is fixed. Besides,
for any root $\alpha$ from $R$ the corresponding reflection $s_{\alpha}$ is the reflection 
with respect to the  hyperplane  $ (v+\rho,\alpha)=0$. Indeed
$$
s_{\alpha}\circ( v)=s_{\alpha}(v+\rho)-\rho=v-2\frac{(v+\rho,\alpha)}{(\alpha,\alpha)}\alpha
$$
So we see if $(v+\rho,\alpha)=0$ then $s_{\alpha}\circ( v)=v$ and $\alpha$ is orthogonal to hyperplane $ (v+\rho,\alpha)=0$.  

Let us define for any $\alpha\in R^+_1$  the following  linear functions on $V$
$$
l^{+}_{\alpha}(v)=(v+\rho,\alpha)-\frac12(\alpha,\alpha),\quad l^-_{\alpha}(v)=(v+\rho,\alpha)+\frac12(\alpha,\alpha)
$$
\begin{definition} A polynomial $f\in S(V^*)$ is called  a quasi-invariant  with respect to  the affine action of the group $G$ if it satisfies the following conditions

$1)$ $f(s_{\alpha}\circ v)=f(v)$ if $\alpha\in R_0$.

$2)$ $f(s_{\alpha}\circ v)-f(v)\in (l^{-}_{\alpha}l^+_{\alpha})$ if $\alpha\in R^+_1$.
\end{definition}
\begin{thm}\label{hc} If $k$ is not rational then
 Harish-Chandra homomorphism is injective and its image is the sub-algebra
$\Lambda^*_{n,m} \subset S(V^*)$
consisting of polynomials which are quasi-invariant with respect  to the affine action of  the group $G$.
\end{thm}
\begin{proof}  Let us show that this Theorem is actually a reformulation of the Theorem 2.2 from \cite{Symmetric}.

Take any $\alpha\in R^+_1$ and  set $h(v)=f(s_{\alpha}\circ v)-f(v)$ and suppose that  $f$ satisfies the conditions 
$$
f(w(v+\rho)-\rho)=f(v), \quad w \in S_n\times S_m
$$
and 
for every $\alpha\in R_1^+$\, $f(v-\alpha)=f(v)$
on the hyperplane $(v+\rho, \alpha)=\frac12(1+k).$
 If $l^+_{\alpha}(v)=0$, then $s_{\alpha}\circ(v)=v-\alpha$, therefore $h(v)=0$. 
Therefore $h\in (l^+_{\alpha})$. 
It is easy to see that $h(s_{\alpha}\circ v)=-h(v)$ and 
$l^+_{\alpha}(s_{\alpha}\circ v)=l^-_{\alpha}(v)$. 
Therefore $h\in (l^-_{\alpha})$ and the conditions of Theorem 2.2 
from \cite{Symmetric} imply the conditions of the present Theorem. 

Now let us prove the opposite statement. So, let $h\in (l^{-}_{\alpha}l^+_{\alpha})$ and $l^+_{\alpha}(v)=0$. Then as we have already seen $s_{\alpha}\circ(v)=v-\alpha$, therefore $0=h(v)=$ $f(v-\alpha)-f(v)$.
\end{proof}
\begin{corollary} Operators $\mathcal L_r$ commute with each other.
\end{corollary}

Now we are going to investigate  the action of the algebra of integrals on the space of  
quasi-invariants.
\begin{definition} 
Denote by $\Lambda^{\pm}_{n,m}$ the subalgebra of quasi-invariants in the algebra $\hat S(V)$ 
$$
\Lambda^{\pm}_{n,m}=\{f\in \hat S(V)^{S_n\times S_m}\mid \partial_{\alpha}f\in (e^{\alpha}-1),\,\alpha\in R_1\}
$$
\end{definition}
It is easy to check that the algebra  $\Lambda^{\pm}_{n,m}$ can be identified with 
the algebra of $S_n\times S_m$-invariant Laurent polynomials
$f\in\Bbb C[x_1^{\pm1},\dots,x_{m+n}^{\pm1}]$
satisfying the  conditions
\begin{equation}
\label{quasix}
x_i\frac{\partial f}{\partial x_i}-kx_j\frac{\partial f}{\partial _j}= 0,\,1\le i\le n,\,n+1\le j\le n+m
\end{equation}
on the hyperplane $x_i=x_j$ for all $i=1,\dots,n$, $j=1,\dots,m$. 

For any Laurent polynomial 
$$
f=\sum_{\mu \in X(n,m)}c_{\mu}x^{\mu}, \quad X(n,m)=\mathbb Z^n\oplus\mathbb Z^m
$$
consider the set $M(f)$ consisting of $\mu$ such that $c_{\mu}\ne0$ and define the {\it support} $S(f)$ as the intersection of the convex hull of $M(f)$ with $X(n,m).$
\begin{thm}\label{in}  The operators  ${\mathcal{L}}_{r}$ for all
$r=1,2,\dots$ map the algebra  $\Lambda^{\pm}_{n,m}$ to itself and  preserve the support: for any $ D\in\mathcal{D}_{n,m}$  and $f\in \Lambda^{\pm}_{n,m}$
$$
S(Df)\subseteq S(f).
$$
\end{thm}
\begin{proof} Let us prove the  first statement.  Let $f\in \Lambda^{\pm}_{n,m}$. Consider the following quasi-homomorphism 
$$
\varphi : V\longrightarrow \Bbb C[x_1^{\pm1},\dots,x_{m+n}^{\pm1}],\, \varphi(\varepsilon_i)=f,\,i=1,\dots,n+m
$$
Therefore by the Theorem \ref{quasi1} $\psi_r=L^r\varphi$ is a quasi-homomorphism and 
$$
e^*L^ref=\psi_r(v),\,v=\varepsilon_1+\dots+\varepsilon_n+\frac1k(\varepsilon_{n+1}\dots+\varepsilon_{n+m})
$$
Therefore $e^*L^ref=\mathcal L_r(f)$ is a quasi-invarant  since the vector $v$ is invariant with respect to the group $G$.

Now let prove the second statement. Consider the quasi-homomorphism  $ \varphi$ such, that $\varphi(\varepsilon_i)=f_i$ where $S(f_i)\subset S(f)$  and $f$ is any quasi-invariant.  Therefore   by Theorem \ref{quasi1}  for any $i,j$ 
$$
g_{ij}=\frac{x_{i}}{x_{i}-x_{j}}(f_i-f_j)
$$
is a polynomial. Since 
$$
f_i-f_j=\left(1-\frac{x_j}{x_i}\right)g_{ij}(x)
$$
and the support of a product of two Laurent polynomials is the Minkowski sum of the supports of the factors  this implies that $S(g_{ij})\subseteq S(f_i-f_j) \subseteq S(f)$. Therefore if $\psi=L\varphi$, then $S(\psi(\varepsilon_i))\subset S(f).$ Now let $\varphi$ be the homomorphism such that $\varphi(\varepsilon_i)=f$. Then by induction on $r$ 
$$
S(\mathcal L_r(f))=S(e^*L^ref)\subset S(f)
$$
\end{proof}
Now we are going to investigate the spectral decomposition of the action of the algebra of CMS integrals $\mathcal D_{n,m}$ on $\Lambda^{\pm}_{n,m}$. 

We will need the following {\it partial order} on the set of integral weights $\chi \in X(n,m)$: we say that
$\tilde\chi \preceq \chi$ if and only if
\begin{equation}
\label{partial}
\tilde\chi_1 \le \chi_1, \, \tilde\chi_1+\tilde\chi_2 \le \chi_1+\chi_2, \dots, \tilde\chi_1+\dots +\tilde\chi_{n+m} \le \chi_1+\dots +\chi_{n+m}
\end{equation}
In the following Proposition we are considering the set $X(n,m)$ as a subset of the vector space $V$ with respect to the  following inclusion
$$
\chi=(\chi_1,\dots,\chi_{n+m})\mapsto \sum_{i=1}^{n+m}\chi_i\varepsilon_i
$$
\begin{proposition}\label{sup} 
Let $f \in \Lambda^{\pm}_{n,m}$ and $\chi$ be a maximal element of $M(f)$ with respect to partial order. 

$1)$ Then for any $D \in \mathcal D_{n,m}$ there is no $\tilde\chi$ in the $M(D(f)), \, \tilde\chi \ne \chi$ such that $\chi \preceq \tilde\chi.$ 
The coefficient at $x^{\chi}$ in $D(f)$ is $\phi_{n,m}(D)(\chi)c_{\chi}\,$, 
where $c_{\chi}$ is the coefficient at $x^{\chi}$  in $f$.

 $2)$ If $\chi$ is the only maximal element of $M(f)$ then $\tilde\chi \preceq \chi$ for any  $\tilde\chi$ from $M(D(f)).$
\end{proposition} 
\begin{proof} Let us prove the first statement.  Let  $\psi$ be a quasi-homomorphism and $\tilde\psi=L\psi$. We are going to prove that  if for any $\tilde\chi\in M(\psi(e_i)),\,i=1,\dots,n+m$  
the inequality $\tilde\chi>\chi$ is impossible, then the same is true for $\tilde\psi$
instead of $\psi$.  If we set 
$$
g_{ij}=\frac{x_i}{x_i-x_j}\left(\psi(\varepsilon_i)-\psi(\varepsilon_j)\right)
$$
then it is easy to verify that if $\tilde\chi\in M(g_{ij})$ then the inequality $\tilde\chi>\chi$ is impossible. 
Since 
\begin{equation}\label{sum1}
\tilde\psi(\varepsilon_i)=\partial_{\varepsilon_i}\psi(\varepsilon_i)+\sum_{j\ne i}g_{ij}
\end{equation}
the same statement is true for $\tilde\psi(\varepsilon_i)$.

Now  let us define a functional  on the space of Laurent polynomials by the formula 
$l_{\chi}(f)=c_{\chi}$, where  $c_{\chi}$ is the coefficient  at $x^{\chi}$ in $f$. Let us prove that 
 $
 l_{\chi}(\tilde\psi)=\phi_{n,m}(L)l_{\chi}(\psi)
 $.
 Since $l_{\chi}$ is a linear functional it is enough to prove that  for every summand in the sum (\ref{sum1}).
But 
 $$
 l_{\chi}(\partial_{\varepsilon_i}f)=k^{p(i)}\chi_il_{\chi}(f)=k^{p(i)}\varepsilon_i^*(\chi)l_{\chi}(f)=\phi_{n,m}(\partial_{\varepsilon_i})(\chi)l_{\chi}(f).
 $$
Since for any $\tilde\chi\in M(\psi(\varepsilon_i))\cup M(\psi(\varepsilon_i))$  the inequality $\tilde\chi>\chi$ is impossible, therefore for $i<j$ we have 
$$
 l_{\chi}(g_{ij})=\begin{cases}l_{\chi}(\psi(\varepsilon_i)-\psi(\varepsilon_j)),\,i<j\\
 0,\,\, i>j
 \end{cases}
$$
This proves the first part. The proof of the second part is similar.
\end{proof}
Let $\theta: \mathcal D_{n,m} \rightarrow \mathbb C$ be a homomorphism and define the corresponding {\it generalised eigenspace} $\Lambda^{\pm}_{n,m}(\theta)$ as the set of all $f\in \Lambda^{\pm}_{n,m}$ such that for every $D\in \mathcal D_{n,m}$ there exists $N\in \mathbb N$ such that $(D-\theta(D))^N(f)=0.$
 If the dimension of $\Lambda^{\pm}_{n,m}(\theta)$ is finite then such $N$ can be chosen independent on $f.$ 
 \begin{proposition}\label{dec} Algebra $\Lambda^{\pm}_{n,m}$  as a module  over the algebra $\mathcal D_{n,m}$ can be decomposed into direct sum of generalised eigenspaces
\begin{equation}
\label{sum}
\Lambda^{\pm}_{n,m}=\oplus_{\theta}\Lambda^{\pm} _{n,m}(\theta),
\end{equation}
where the sum is taken over the set of some homomorphisms $\theta$ (explicitly described below).
\end{proposition}
 \begin{proof} 
 Let $f\in \Lambda^{\pm}_{n,m}$ and define a vector space
$$
W(f)=\{g\in \Lambda^{\pm}_{n,m}\mid S(g)\subseteq S(f)\}.
$$
By Theorem \ref{in} $W(f)$ is a finite dimensional module over $\mathcal D_{n,m}.$ 
Since the proposition is true for all finite-dimensional modules, the claim follows.
\end{proof}

Now we describe all homomorphisms $\theta$ such that $\Lambda^{\pm}_{n,m}(\theta)\ne0$.
We say that the integral weight $\chi \in X(n,m)$ {\it dominant} if 
$$
\chi_1 \ge \chi_2 \ge \dots \ge \chi_n, \quad \chi_{n+1} \ge \chi_{n+2} \ge \dots \ge \chi_{n+m}.
$$
The set of dominant weights is denoted by $X^+(n,m).$

For every $\chi\in X^+(n,m)$  we define the homomorphism $\theta_{\chi}: \mathcal D_{n,m} \rightarrow \mathbb C$ by
$$
\theta_{\chi}(D)=\phi_{n,m}(D)(\chi), \,\, D\in \mathcal D_{n,m}
$$
where $\phi_{n,m}$ is the Harish-Chandra homomorphism. 
\begin{proposition}\label{max}
$1)$ For any $\chi\in X^+(n,m)$ there exists $\theta$ and $f_{\chi}\in \Lambda^{\pm}_{n,m}(\theta)$, which has the only maximal term $x^{\chi}$.

$2)$  $\Lambda^{\pm}_{n,m}(\theta)\ne0$ if and only if there exists $\chi\in X(n,m)^+$ such that $\theta=\theta_{\chi}$.

$3)$ If  $\Lambda^{\pm}_{n,m}(\theta)$  is finite dimensional then its dimension is equal to the number of $\chi \in X^+(n,m)$ such that $\theta_{\chi}=\theta$.
\end{proposition}

For the proof see \cite{Symmetric}.

\begin{corollary}
The set of homomorphisms in Proposition \ref{dec} consists of $\theta=\theta_{\chi}, \,\, \chi \in X^+(n,m).$
\end{corollary}

\section{Description of equivalence classes} 
Following the last statement   of the Proposition \ref{max} let us introduce the following definition.
\begin{definition} Two weights $\chi,\tilde\chi \in X^{+}(n,m)$ are called equivalent if $\theta_{\chi}=\theta_{\tilde\chi}$.
\end{definition}
In this section we are going to investigate further  for non rational $k$ this equivalence relation.
 In particular we will describe explicitly  the equivalence classes. It is easy to see from Theorem \ref{hc} that  two weights  $\chi,\tilde\chi\in X^+(n,m)$ are equivalent if and only if $f(\chi)=f(\tilde\chi)$ for any $f\in \Lambda^*_{n,m}$. Let denote by $E(\chi)$ the equivalence class containing $\chi$. In the paper \cite{Weyl} we described equivalence classes without any assumptions on the weight $\chi$. But in the case when $\chi \in X^+(n,m)$ and $k$ is not rational we can say much more. It turns out that in such a case  it is more convenient to use  a geometric language of polygonal lines, rather then the one of  the Young diagrams.

\begin{definition} Let $a_1,\dots,a_n$ be any non-increasing
sequence of integers. Consider  $2n$ points on the plane 
$$
M_1=(a_1,0), M_2=(a_1,1), \dots, M_{2n-1}=(a_n,n-1),\, M_{2n}=(a_n,n)
$$
and two additional points ``at infinity'',
$$
M_0=(+\infty,0),\, M_{2n+1}=(-\infty,n).
$$
Let us  denote by $\Gamma_a$ the polygonal line
$$
\Gamma_a=\bigcup_{i=0}^{2n}[M_i,M_{i+1}]
$$
consisting of the  segments $[M_i,M_{i+1}]$.
The  segments $[M_0,M_{1}],[M_{2n},M_{2n+1}]$ mean the corresponding half lines.
\end{definition}   
Let us  also denote by $\tau_n$ and $\omega$ the following transformations  of $\Bbb R^2$
$$
\tau_n(x,y)=(x,y+n),\quad \omega(x,y)=(y,x)
$$
Let $\chi \in X^+(n,m)$. Then we can write $\chi=(a_1,\dots,a_n\mid b_1,\dots, b_m)$  where $a_1,\dots,a_n$ and $ b_1,\dots, b_n$ are two sequences of non increasing integers. Then we can define two polygonal lines: $\Gamma_a$ and  $\hat{\Gamma}_b=\tau_n\circ\,\omega(\Gamma_b)$. It easy to check that the polygonal line $\hat{\Gamma}_b$ can be described in the following way 
$$
\hat{\Gamma}_b=\bigcup_{j=0}^{2m}[N_j,N_{j+1}]
$$
where 
$$
 N_1=(0,n+b_1),\,N_2=(1,n+b_1),\dots,N_{2m-1}=(m-1,n+b_m),\,N_{2m}=(m,n+b_m)
$$
and
$$
N_0=(0,+\infty),
N_{2m+1}=(m, -\infty).
$$

\begin{definition}
Denote by $D^{-}_a$ the part of the plane $\Bbb R^2$ bounded 
by the lines $\Gamma_a,\, y=n,\,x=0$.
Denote by $D_a^+$ the part of the plane bounded by the lines  $\Gamma_a,\,y=0,\, x=0$.
For $\chi=(a_1,\dots,a_n\mid b_1,\dots, b_m)$  define 
$$
\hat\Gamma^-_b=\tau_n\circ\omega(\Gamma_b^-),\quad \hat\Gamma^+_b=\tau_n\circ\omega(\Gamma_b^+)
$$
\end{definition}

Now let us set 
$$
b_r(x_1,\dots,x_n,k,h)=\sum_{i=1}^n\left[ B_r(x_i+k(i-1)+h)-B_r(k(i-1)+h)\right]
$$
where $B_r(x)$ are the Bernoulli  polynomials and $k,h$ are complex numbers.
\begin{proposition} For any  $r\in \Bbb Z_{>0}$  the polynomials 
$$
b^{(n,m)}_{r}(\xi)=b_r(\xi_1,\dots,\xi_n,k,0)+k^{r-1}b_r(\xi_{n+1},\dots,\xi_{n+m},k^{-1},n)
$$
belong to the algebra $\Lambda^*_{n,m}$  and generate it.
\end{proposition}
\begin{proof} We see that $b^{(n,m)}_r(\xi)=f_r(\xi+\rho+v)$, where 
$$
v=\frac12(kn+m-k)\sum_{i=1}^n\varepsilon_i+\frac12(n+k^{-1}m-k^{-1})\sum_{j=1}^m\varepsilon_{n+j}
$$
and $f_r$ is symmetric with respect to $S_n\times S_m$. Therefore for any even root $\alpha$ we have 
$$
b^{(n,m)}_r(s_{\alpha}\circ\xi)=b^{(n,m)}_r(s_{\alpha}(\xi+\rho)-\rho)=f_r(s_{\alpha}(\xi+\rho)+v)=f_r(s_{\alpha}(\xi+\rho+v))
$$
$$
=f_r(\xi+\rho+v)=b_r(\xi)
$$
Now let $\alpha=\varepsilon_i-\varepsilon_{n+j}$ be an odd root and $(\xi+\rho,\alpha)=\frac12(\alpha,\alpha)$, so we have
$$
b^{(n,m)}_r(\xi)-b^{(n,m)}_{r}(\xi-\alpha)=B_r(\xi_i+k(i-1))+k^{r-1}(\xi_{n+j}+k^{-1}(j-1)+n)$$
$$
-B_r(\xi_i-1+k(i-1))-k^{r-1}B_r(\xi_{n+j}+1+k^{-1}(j-1)+n)
$$
Since the  Bernoulli polynomials have the property
$
B_r(x+1)-B_r(x)=rx^{r-1}
$
then we have 
$$
b^{(n,m)}_r(\xi)-b^{(n,m)}_{r}(\xi-\alpha)=r(\xi_i-1+k(i-1))^{r-1}-rk^{r-1}(\xi_{n+j}+k^{-1}(j-1)+n)^{r-1}
$$
But it is not difficult to verify, that condition 
$$
\xi_i-1+k(i-1)=k(\xi_{n+j}+k^{-1}(j-1)+n)
$$
is equivalent to the condition $(\xi+\rho,\alpha)=\frac12(\alpha,\alpha)$.
Therefore the polynomials $b^{(n,m)}_r(\xi)$ belong to the algebra $\Lambda^*_{n,m}$. The fact that they generate this algebra has been proved in \cite {Gendis}. 
\end{proof}

\begin{corollary} Let $\chi,\tilde\chi\in X^+(n,m)$. Then they are equivalent  if and only if 
$$
b^{(n,m)}_r(\chi)=b^{(n,m)}_r(\tilde\chi),\,\, r=1,2,\dots.
$$
\end{corollary}
\begin{definition} Let $\square$ be a  a unit  square  on the plane such that 
all its vertices have integer coordinates. Then we denote $c_k(\square)=x+ky$, 
where $k$ is a complex number and $(x,y)$  are the coordinates of the left lower vertex. 
\end{definition}
The next lemma gives some geometric formula for the value  of $b^{n,m}_r(\chi)$.
\begin{lemma}\label{formula} Let $\chi\in X^+(n,m)$ then 
$$
b^{(n,m)}_r(\chi)=r\left[\sum_{\square\in D^+_a\cup \hat D^+_b}c_{k}(\square)^r-\sum_{\square\in D^-_a\cup \hat D^-_b}c_{k}(\square)^r\right]
$$
\end{lemma}
\begin{proof}
 From the definition of the Bernoulli polynomials we have the following formulae for integer $z$
$$
B_r(z+h)-B_r(h)=\begin{cases}\sum_{j=1}^z(j-1+h)^{r-1},\,z>0\\
\sum_{j=z}^{-1}(j+h)^{r-1},\,z<0\end{cases}
$$
Therefore 
$$
\sum_{a_i >0}B_r(a_i+k(i-1)+h)^{r-1}-B_r(k(i-1)+h)^{r-1}=
r\sum_{\square\in D^+_a}(c_{k}(\square)+h)^{r-1}
$$
and 
$$
\sum_{a_i <0}B_r(a_i+k(i-1)+h)^{r-1}-B_r(k(i-1)+h)^{r-1}=
-r\sum_{\square\in D^-_a}(c_{k}(\square)+h)^{r-1}
$$
Further we have  
$$
k^{r-1}\sum_{j=1}^mB_r(b_i+ k^{-1}(j-1)+n)-B_r(k^{-1}(j-1)+n)=
$$
$$
rk^{r-1}\sum_{\square\in D^+_b}(c_{k^{-1}}(\square)+n)^{r-1}-rk^{r-1}\sum_{\square\in D^-_b}(c_{k^{-1}}(\square)+n)^{r-1}
$$
and we only need to show that
$$
k^{r-1}(c_{k^{-1}}(\square)+n)=c_k(\tau_n\circ\omega(\square))
$$
But it easy follows from the direct calculations and the fact that if $(x,y)$ is the left lower vertex of a square $\square $ then $(y,x+n)$ is the left lower vertex of the square $\tau_n\circ\omega(\square)$.
\end{proof}
\begin{corollary}\label{formula1}  Let us set 
$$
D_{\chi}^+=D^+_a\cup\hat D^+_{ b},\quad D_{\chi}^-=D^-_a\cup\hat D^-_{ b}\,.
$$
Then the following formulae hold true:
$$
b^{(n,m)}_r(\chi)=r\left[\sum_{\square\in D^+_{\chi}\setminus D^-_{\chi}}c_{k}(\square)^{r-1}-\sum_{\square\in D^-_{\chi}\setminus D^+_{\chi}}c_{k}(\square)^{r-1}\right],\,r=1,2,\dots
$$
\end{corollary}

\begin{thm} 
Suppose  $k\notin\Bbb Q$ and 
$$
\chi=(a_1,\dots,a_n\mid b_1,\dots,b_m),\quad \tilde\chi=(\tilde a_1,\dots,\tilde a_n\mid \tilde b_1,\dots,\tilde b_m).
$$
Then the conditions $b^{(n,m)}_r(\chi)=b^{(n,m)}_r(\tilde\chi)$ for 
$k=1,2,\dots$ are equivalent to the condition 
$$
\Gamma_a\cup\hat\Gamma_b=\Gamma_{\tilde a}\cup\hat\Gamma_{\tilde b}
$$
\end{thm}

\begin{proof}  It is easy to see that  equality  $
\Gamma_a\cup\hat\Gamma_b=\Gamma_{\tilde a}\cup\hat\Gamma_{\tilde b}$ is equivalent to the  equalities  
 $$
D^+_{\chi}\setminus D^-_{\chi}=D^+_{\tilde\chi}\setminus D^-_{\tilde\chi},\quad D^-_{\chi}\setminus D^+_{\chi}= D^-_{\tilde\chi}\setminus D^+_{\tilde\chi}
$$
So  since  $k$ is not a rational number the  last two equalities are equivalent to  the statement that  two sequences 
$$
(c_k(\square))_{\square\in D^+_{\chi}\cup D^-_{\tilde\chi}},\quad (c_k(\square))_{\square\in  D^+_{\tilde\chi}\cup  D^-_{\chi}}
$$
coincide up to a permutation. And finally  by Corollary \ref{formula1} this  is equivalent to the conditions
$b^{(n,m)}_r(\chi)=b^{(n,m)}_r(\tilde\chi), k=1,2,\dots$.
\end{proof}

Let $\chi\in X^+(n,m)$. Consider the decomposition of the intersection  $\Gamma_a\cap\hat \Gamma_b$ into connected components
$$
\Gamma_a\cap\hat \Gamma_b=\gamma_1\cup\dots\cup \gamma_{r+1}
$$

And let $P_i,Q_i$ be the boundary points of the line $\gamma_i$  (if  $\gamma_i$ is one point then $P_i=Q_i$).
We suppose that
$$
P_1\ge Q_1\ge P_2\ge Q_2\ge\dots\ge P_{r+1}\ge Q_{r+1}
$$
and we use here the total order on the points such that $P=(x,y)> \tilde P=(\tilde x,\tilde y)$ if and only if 
$y>\tilde y$ or $y=\tilde y$ and $x<\tilde x$. For every $i=1,\dots,r$
there are exactly two ways to get from   
$Q_i$ to $P_{i+1}$ along  the line $\Gamma_a\cup\hat \Gamma_b$. 
Let us denote  the lower way by $L_i$  and and the upper way by $U_i$. 
Denote by $\nu_i$  the part of the plane
bounded by $L_i$ and $U_i$. Set $\nu=\cup_{i=1}^r\nu_i$.

\begin{corollary}\label{curves} Let $\chi\in X^+(n,m)$  then:  

$1)$ If  $\tilde\chi \in E(\chi)$  then $\Gamma_{\tilde a}$  can be obtained from $\Gamma_a$ by replacing one of the   two possible ways  from $Q_i$ to $P_{i+1}$ by the other one for some of the indices $i=1,\dots,r.$

$2)$ The number of elements in $E(\chi)$  is equal to  $2^r$.

$3)$ Every equivalence class contains a unique  weight $\chi_{min}$ 
such that  the corresponding  line $\Gamma_{a_{min}}$ contains  
all the lower paths between the points $P_i$ and  $Q_{i+1}$  for each $i=1,\dots,r$.

\end{corollary}
\begin{proof} This corollary easily follows from geometric considerations.
\end{proof}

Now  we are going to reformulate the previous results  in terms of  the weights, odd roots
and  the deformed scalar product. Let us define the bijection 
 $
 \eta:R_1^+\longrightarrow Q_{m,n}
 $ 
from the set of odd positive roots to the set of unit squares with integer coordinates 
contained in the rectangle $[0,m]\times[0,n]$ by the following rule: 
$\eta(\varepsilon_i-\varepsilon_{n+j})$ is the square with the upper 
right vertex $(j,i)$. The following theorem  is the main result of this section.

\begin{thm}\label{desc} Let  $E$ be an equivalence class and $\chi=\chi_{min}\in E$. Then the following statements hold true

$1)$  We have $\mid E\mid=2^r$, then $r$ is equal to the number of the odd positive roots $\alpha$  such that  
$$
(\chi+\rho,\alpha)+\frac12(\alpha,\alpha)=0
$$

$2)$ Let us denote by $R(\chi)$ the set of all $\alpha\in R_1^+$ such that there exist a sequence of odd positive roots $\alpha_1,\dots,\alpha_N$ satisfying the following conditions
\begin{equation}\label{conection}
(\chi+\rho+\alpha_1+\dots+\alpha_{i-1},\alpha_i)+\frac12(\alpha_i,\alpha_i)=0,\,i=1,\dots,N
\end{equation}
Then $\eta(R(\chi))=\nu$.

$3)$ If 
$$
R(\chi)=R^{(1)}\cup\dots\cup R^{(s)}
$$ 
is the decomposition into orthogonal  components, then $s=r$ and
$$
\eta(R(\chi))=\eta(R^{(1)})\cup\dots\cup \eta(R^{(s)})
$$
is the decomposition into connected  components.

$4)$
Set $\beta_t=\sum_{\alpha\in R^{(t)}}\alpha$. Then every weight $\tilde\chi$  
from $E(\chi)$ can be written as
$$
\tilde\chi=\chi+\vartheta_1\beta_1+\dots+\vartheta_r\beta_r,\,\,\vartheta_t\in\{0,1\},\,t=1,\dots,r
$$

$5)$  There exist pairwise commuting elements 
$g_1,\dots,g_r$  of the group $G$ such that any element $\tilde\chi\in E(\chi)$ 
can be written as
$$
\tilde\chi=g^{\vartheta_1}_1\circ g^{\vartheta_2}_2\circ\dots\circ g^{\vartheta_r}_r(\chi_{\min})
$$

$6)$ Let $\tilde\chi\in E(\chi)$. Then 
$$
\prod_{\alpha\in R^+_1}[(\tilde\chi+\rho,\alpha)-\frac12(\alpha,\alpha)]\ne0
$$
if and only if $\tilde\chi=\chi_{min}$.

\end{thm}
\begin{proof} Let us prove the first statement. It is easy to check that for $\alpha=\varepsilon_i-\varepsilon_{n+j}$ we have 
$$
(\chi+\rho,\alpha)+\frac12(\alpha,\alpha)= \chi_i-j+1-k(\chi_{n+j}+n-i)
$$
Therefore the condition $(\chi+\rho,\alpha)+\frac12(\alpha,\alpha)=0$ is equivalent to the conditions
$$
\chi_i-j+1=0,\quad \chi_{n+j}+n-i=0
$$
But it is easy to check that the last two conditions are equivalent to the condition 
$M_{2i}=N_{2j-1}$.  Since  $\chi=\chi_{\min}$ then there exists $1\le s\le r$ such that   
$M_{2i}=N_{2j-1}=Q_s$ and  such $s$ is unique. And it is easy to check that this correspondence is a bijection between the points $Q_{s}, 1\le s \le r$ and the set of $\alpha\in R^+_1$ such that $(\chi+\rho,\alpha)+\frac12(\alpha,\alpha)=0$.
Thus we proved the first statement.

Let us prove the second statement. 
Let $\square \in Q_{n,m}$  be the square with the
upper right vertex $(j,i)$.  Then we set 
$$
f(\square)=a_i-j+1,\,g(\square)=b_j+n-i
$$
It is easy to check the following statements. 

$a)$ If $\square$ is  located  on the right of $\Gamma_a$ then $-f(\square)$ is equal to  the number of cells between the cell $\square$ and the line $\Gamma_a$ with the same coordinate $i$; if  $\square$ is located  on the left of $\Gamma_a$ then $f(\square)-1$ is equal to  the number of cells between the cell $\square$ and the line $\Gamma_a$ with the same coordinate $i$.

$b)$  If $\square$  is located below $\hat\Gamma_b$ then $g(\square)$ 
is equal to  the number of cells between the cell $\square$ and the line $\hat\Gamma_b$ 
with the same coordinate $j$;  if $\square$ located  above $\hat\Gamma_b$ 
then $-1-g(\square)$ is equal to  the number of cells between the cell $\square$ and the line $\hat\Gamma_b$ with the same coordinate $j$.

Let us first prove that $R(\chi)\supset \eta^{-1}(\nu)$. 
Let  $\alpha=\varepsilon_i-\varepsilon_{n+j}$  and $\eta(\alpha)\in \nu$. 
Let $\nu_t$ be the connected component containing $\alpha$.
Let $\{\square_1,\dots,\square_N\}$ be a sequence of cells  
located on the left of or above the cell $\eta(\alpha)$ in the component $\nu_t$. 
Then it is easy to verify that 
$$
(\eta(\square_1)+\dots+\eta(\square_N),\alpha)=A+kB
$$
where $A$ is the number of  of cells between the cell $\eta(\alpha)$  and the line $\Gamma_a$  with the same coordinate $i$  and $B$ is the number  of cells between the cell $\eta(\alpha)$ and the line $\hat\Gamma_b$ with the same coordinate $j$. Therefore we have 
$$
(\chi+\rho+\alpha_1+\dots+\alpha_{N},\alpha)+\frac12(\alpha,\alpha)=0,
$$
and we have proved the inclusion $R(\chi)\supset \eta^{-1}(\nu)$. Let us prove the opposite inclusion. 
Suppose that the conditions  (\ref{conection})  are fulfilled. Then 
$$
(\alpha_1+\dots+\alpha_N,\alpha)=-(\chi+\rho,\alpha)-\frac12(\alpha,\alpha)=-f(\eta(\alpha))+kg(\eta(\alpha))
$$
Therefore $f(\eta(\alpha))\le0$ and $g(\eta(\alpha))\ge 0$. So $\eta(\alpha)\in\nu$ and the second statement is proved.

To prove the third  statement let us consider the decomposition of $\nu$ into connected components
$
\nu=\nu_1\cup\nu_2\dots\cup\nu_r. 
$
Then 
$$
R(\chi)=\eta^{-1}(\nu_1)\cup\eta^{-1}(\nu_2)\dots\cup\eta^{-1}(\nu_r) 
$$
and it is easy to see that $\eta^{-1}(\nu_t)\bot\eta^{-1}(\nu_s)$ for $t\ne s$. Since the set   $\eta^{-1}(\nu_t)$ (as a set of cells) is connected  it can not be represented as a disjoint union of orthogonal subsets. And the third statement is proved.

Now let us prove the fourth  statement. Let $\tilde\chi\in E(\chi)$. 
Then according to the Corollary \ref{curves}, there exist a subset $D\subset \{1,2,\dots,r\}$ such that 
$\Gamma_{\tilde a}$ can be obtained from $\Gamma_{a}$ by replacing  $L_i$ by $U_i$  and   $\hat\Gamma_{\tilde b}$ can be obtain from $\hat \Gamma_{ b}$ by replacing $U_i$ by $L_i$ for $i\in D$. It is enough to consider the case when  $D=\{ t\}$ consists of one element. Let  $\nu_t$ be the corresponding connected component. Then we can define two sets
$$
I_t=\{i\in\{1,\dots.n\}\mid \exists \, j,\,\varepsilon_i-\varepsilon _{n+j}\in \eta^{-1}(\nu_t)\}
$$ 
and 
$$
J_t=\{j\in\{1,\dots.n\}\mid\exists\,i,\, \varepsilon_i-\varepsilon _{n+j}\in \eta^{-1}(\nu_t)\}
$$ 
Since the connected component $\nu_t$  is a skew Young diagram  it has rows and columns. For $i\in I_t$  let us denote by $d_i$ the length of the row of the skew diagram $\nu_t$ which contains box $\eta(\varepsilon_i-\varepsilon_{n+j})$ for some $j$. Similarly for $j\in I_t$  let us denote by $\tilde d_j$ the length of the column  of the skew diagram $\nu_t$ which contains box $\eta(\varepsilon_i-\varepsilon_{n+j})$ for some $i$. Then it is not difficult to verify the following formulae
$$
\tilde a_i=\begin{cases}a_i,\,i\notin I_t\\
a_i+d_i,\,i\in I_t,
\end{cases}\,1\le i\le n\quad
\tilde b_j=\begin{cases}b_j,\,j\notin J_t\\
b_j-\tilde d_j,\,j\in J_t,
\end{cases}\,1\le j\le m
$$
Therefore
$$
\tilde\chi=\chi+\sum_{i=1}^pd_i\varepsilon_i-\sum_{j=1}^q\tilde d_j\delta_j
$$
But 
$$
\sum_{i=1}^pd_i\varepsilon_i-\sum_{j=1}^q\tilde d_j\delta_j=\sum_{\square\in\nu_t}\eta^{-1}(\square)=\beta_t
$$
and we proved the fourth statement.

Let us prove the fifth statement. Let 
$$
R^{(s)}=\{\alpha_1,\dots,\alpha_N\}
$$ 
be one of the orthogonal components of $R(\chi)$. Suppose that 
$$
(\chi+\rho+\alpha_1+\dots+\alpha_{i-1},\alpha_i)+\frac12(\alpha_i,\alpha_i)=0,\,i=1,\dots,N
$$
Set $g_s=s_{\alpha_N}\circ\dots\circ s_{\alpha_1}$. Then 
by using the previous relations, it is easy to check that $g(\chi)=\chi+\beta_s$. 
Since $R^{(s)}\perp R^{(t)}$ for $s\ne t$, the elements $g_s$ and $g_t$ commute. 

Let us prove the sixth statement. We will prove first that $(\chi_{\min}+\rho,\alpha)-\frac12(\alpha,\alpha)\ne0$ for any $\alpha\in R_1$.  If $\alpha=\varepsilon_i-\varepsilon_{n+j}$, and  
$$
(\chi_{\min}+\rho,\alpha)-\frac12(\alpha,\alpha)=a_i-j-k(b_j+n-i+1)=0
$$
then $M_{2i-1}=N_{2j}$ is the intersection point of the lines $\Gamma_a$ and 
$\hat\Gamma_b$. But this is impossible  since  $\Gamma_a$ is located 
on the left of the line $\hat\Gamma_b$.
Now let 
$$
\tilde\chi=\chi_{\min}+\vartheta_1\beta_1+\dots+\vartheta_r\beta_r\,.
$$
Then there is $s$ such that $\vartheta_s=1$. 
Let $Q_s=(ji-1)$ be the  intersection point of the lines $\Gamma_{\tilde a}$ and $\hat\Gamma_{\tilde b}$.  
Since $\vartheta_s=1$ the line $\hat\Gamma_{\tilde b}$ contains also  point $(j-1,i-1)$. 
Therefore  $(ji-1)=N_{2j}$ and since the line $\Gamma_{\tilde a}$ contains  point $(j, i)$ we have $(j,i-1)=M_{2i-1}$. Therefore
$$
(\tilde\chi+\rho,\alpha)-\frac12(\alpha,\alpha)=0
$$
and the theorem is proved.
\end{proof}
\begin{definition}\label{reg}
Let us denote by $X_{reg}^+(n,m)$ the set of weights  $\chi$ from $X^+(n,m)$
such that $(\chi+\rho,\alpha)-\frac12(\alpha,\alpha)\ne0$ for any positive odd root $\alpha$. 
\end{definition}
For brevity  we  will write $\Lambda^{\pm}_{n,m}(\chi)$ instead of  
$\Lambda^{\pm}_{n,m}(\theta_{\chi})$. The next proposition is a description of the 
spectral decomposition in terms of the root system.

\begin{proposition}\label{dec1} The following statements hold true:

$1)$ $\Lambda^{\pm}_{n,m}$  as a module  over the algebra 
$\mathcal D_{n,m}$ can be decomposed into a direct sum of generalised eigenspaces
\begin{equation}
\label{sum}
\Lambda^{\pm}_{n,m}=\bigoplus_{\chi\in X_{reg}^+(n,m)}\Lambda^{\pm}_{n,m}(\chi),
\end{equation}
where $\Lambda^{\pm}_{n,m}(\chi)$ is the generalised eigen-space  corresponding to the homomorphism  $\theta_{\chi}$.

$2)$  The dimension  of the space $\Lambda^{\pm}_{n,m}(\chi)$ is equal to $2^r$ where $r$ is the number of the odd positive roots such that $(\chi+\rho,\alpha)+\frac12(\alpha,\alpha)=0$.

$3)$ Algebra $\Lambda^{\pm}_{n,m}$ is generated by the deformed power sums
$$
p_s(x_1,\dots,x_{n+m})=x^s_1+\dots+x_n^s+\frac1k(x^s_{n+1}+\dots+x_{n+m}^s),\,\,s=\pm1,\pm2,\dots
$$
\end{proposition}

\begin{proof} The first two statements follow from Proposition \ref{dec}, Proposition \ref{max} and  the Theorem \ref{desc}. 

So 
let us prove the third statement. 
Let us denote by $\tilde\Lambda^{\pm}_{n,m}$ the subalgebra in $\Lambda^{\pm}_{n,m}$ which is generated by the deformed power sums and  
$$
\Lambda^+_{n,m}=\Lambda^{\pm}_{n,m}\cap\Bbb C [x_1,\dots,x_{n+m}],\,\, \Lambda^-_{n,m}=\Lambda^{\pm}_{n,m}\cap\Bbb C [x^{-1}_1,\dots,x^{-1}_{n+m}]
$$ 
We have already proved  that  subspace $\Lambda^{\pm}_{n,m}(\chi)$ is finite dimensional.  
Let us prove that $\Lambda^{\pm}_{n,m}(\chi)\subset \tilde\Lambda^{\pm}_{n,m}$. 
Let $\chi=(a_1,\dots,a_n\mid b_1,\dots,b_m)\in X^+(n,m)$    and $c,d$ are integers such that    
$$
\lambda_i=a_i-2m-c\ge 0,\,\,i=1,\dots,n,\quad \mu_j=b_j+2n-d\ge 0,\,\,j=1,\dots,m
$$
In other words $\lambda=(\lambda_1,\dots,\lambda_n)$ and $\mu=(\mu_1,\dots,\mu_m)$ are partitions.

Consider the polynomials 
$$
g_{1}(x_1,\dots,x_{n+m})=\prod_{i,j}(x_i-x_{n+j})^2s_{\lambda}(x_1,\dots,x_n)s_{\mu}(x_{n+1},\dots,x_{n+m})
$$
and
$$
g_{2}(x_1,\dots,x_{n+m})=\prod_{i,j}(x_i^{-1}-x_{n+j}^{-1})^2(x_1\dots x_n)^c(x_{n+1},\dots,x_{n+m})^d
$$ 
where $s_{\lambda}(x_1,\dots,x_n),\,s_{\mu}(x_{n+1},\dots,x_{n+m})$ are Schur  polynomials. Since the highest terms of these  Schur polynomials have weights $\lambda$ and $\mu$  the highest term of  the product 
$$
g(x_1,\dots,x_{n+m})=g_1(x_1,\dots,x_{n+m})g_2(x_1,\dots,x_{n+m})
$$  has the weight $\chi$. 
It is  easy to verify that $g_{1}(x_1,\dots,x_{n+m})\in \Lambda^+_{n,m}$ and $g_2(x_1,\dots,x_{n+m})\in \Lambda^-_{n,m}$.
But by Theorem 2 from \cite{Deformed}  the 
algebra $\Lambda^+_{n,m}$ is generated by the positive power sums,
hence the algebra $\Lambda^-_{n,m}$  is generated by the negative power sums. 
So $g(x_1,\dots,x_{n+m})\in \tilde\Lambda^{\pm}_{n,m}$.   
Since the integrals preserves the algebra $\tilde\Lambda^{\pm}_{n,m}$ 
(see \cite{Dunkl} formulae (46), (47))  one can define $\tilde\Lambda_{n,m}^{\pm}(\chi)$ in the same way as $\Lambda_{n,m}^{\pm}(\chi)$ was defined and 
the same arguments as in Proposition \ref{max} show that there exists 
$f_{\chi}\in \tilde\Lambda_{n,m}^{\pm}(\chi)$ with the only maximal weight $\chi$.  
Hence $f_{\tilde\chi},\,\tilde\chi\in E(\chi)$ make a basis in  $\Lambda^{\pm}_{n,m}(\chi)$. 
Therefore $\tilde\Lambda^{\pm}_{n,m}(\chi)=\Lambda^{\pm}_{n,m}(\chi)$. 
\end{proof}

\section{Weights and bipartitions}

By a bipartition  we will mean a pair of partitions $(\lambda,\mu)$. 
We will denote by $H(n,m)$ the set  of partitions  $\lambda$ such that 
$\lambda_{n+1}\le m$, or $\lambda'_{m+1}\le n$.  
The following definition is not standard, but one can easily check that  
it is equivalent  to the usual definition \cite{Moens}. 
 
\begin{definition}\label{cross}  We say that the bipartition $(\lambda,\mu)$ is contained in the $(n,m)$ cross 
if  there are nonnegative integers $p,q,r,s$ such that $p+q=n,\,r+s=m$  
and $\lambda\in H(p,r),\, \mu\in H(q,s)$. We will denote the set  of all  such bipartitions  by  $Cr(n,m)$.
\end{definition} 
 
The main aim of this section is to give a geometric description  of a  bijection 
between  the set  $Cr(n,m)$ and the set $X^{+}(n,m)$. 
 
\begin{definition} Let $(\lambda,\mu)\in Cr(n,m)$.  Set
$$
H(\lambda,\mu)=\{(i,j)\mid \lambda\in H(i,m-j),\,\mu\in H(n-i,j)\}
$$
and define
$$
p=\max\{ i\mid (i,j)\in H(\lambda,\mu)\},\,\, s=\max\{j\mid(p,j)\in H(\lambda,\mu)\}.
$$
We will call the pair $(p,s)$ the extremal  one.
\end{definition}  

There is an easy way to find the extremal pair for a given bipartition $(\lambda,\mu)\in Cr(n,m)$.

\begin{definition} 
Let $(\lambda,\mu)\in Cr(n,m)$. Then we set $F(\lambda,\mu)=(\tilde\lambda,\tilde\mu)$ where 
$$
\tilde\lambda_i=\lambda_i-\lambda_{n+1},\,1\le i\le n+1,\quad \tilde\mu'_j=\mu'_j-\mu'_{m+1},\,1\le j\le m+1
$$
and we also set $\tilde n=n-\mu'_{m+1},\,\tilde m=m-\lambda_{n+1}$.
\end{definition}

\begin{lemma} 
The map $F$ has the following properties:

$1)$ If $(\lambda,\mu)\in Cr(n,m)$ then $F(\lambda,\mu)\in Cr(\tilde n,\tilde m)$

$2)$ If the pair $(p,s)$ is extremal for $(\lambda,\mu)$ then it is also  extremal for $F(\lambda,\mu)$.

$3)$  $F(\lambda,\mu)=(\lambda,\mu)$ if and only if $\lambda\in H(n,0),\,\mu\in H(0,m)$.
The extremal pair in this case is $(n,m)$.

\end{lemma}  
\begin{proof}
 Let us prove the first statement.  
We have 
$$
\tilde\lambda_{p+1}=\lambda_{p+1}-\lambda_{n+1}\le  m-s-\lambda_{n+1}=\tilde m-s
$$
$$
\tilde\mu'_{s+1}=\mu'_{s+1}-\mu'_{m+1}\le n-p-\mu'_{m+1}=\tilde m-p
$$
So we see that $\tilde\lambda\in H(p,\tilde n-s)$ and $\tilde\mu\in H(n-p,s)$, or equivalently 
this means that $(\tilde \lambda,\tilde\mu)\in Cr(\tilde n,\tilde m)$.
This we proved the first statement. 

Now let us prove the second statement. Let $(\tilde p,\tilde s)\in H(\tilde\lambda,\tilde\mu)$. Then by definition we have:
$$
\tilde\lambda_{\tilde p+1}\le \tilde m-\tilde s , \,\,\tilde\mu'_{\tilde s+1}\le \tilde n-\tilde p
$$
Therefore 
$$
\lambda_{\tilde p+1}\le m-\tilde s,\,\,\mu'_{\tilde s+1}\le n-\tilde p
$$ 
and $\tilde p\le p$ and $\tilde s\le s$. This proves the second statement. 
The third  statement easily follows from the definition.
\end{proof}

There exists a natural map $\pi : Cr(n,m)\longrightarrow X^{+}(n,m)$ (see for example \cite{Moens}) and we are going to give a geometric interpretation of this map.

\begin{definition} 
Let $(\lambda,\mu)\in Cr(n,m)$ and $(p,s)$ be corresponding extremal pair. 
Then we set
$$
\pi(\lambda,\mu)=(\lambda_1,\dots,\lambda_p,m-\mu_q,\dots,m-\mu_1\mid\lambda'_1-n,\dots,\lambda'_r-n,-\mu'_s,\dots,-\mu'_1)
$$
\end{definition}
It is not difficult to verify that 
$\pi(\lambda,\mu)\in X^{+}_{n,m}$
for any $(\lambda,\mu)\in Cr(n,m)$.

Below we are going to prove some properties of the map $\pi$ by
using our geometric interpretation of the set  $X^+(n,m)$ and 
a geometric interpretation of bipartitions, in the  same way
as S. Kerov did in his paper \cite{Kerov} for partions.

\begin{definition} Let $\lambda=(\lambda_1, \lambda_2,\dots)$ be a partition. 
Consider  the following points on the plane:
$$
M_0=(+\infty,0),\, 
M_1=(\lambda_1,0),\,
M_2=(\lambda_1,1),\,
M_3=(\lambda_2,1),\,
M_4=(\lambda_2,2),\,
\dots\,.
$$
Let us denote by $Y_{\lambda}$ the polygonal line
$$
Y_{\lambda}=\bigcup_{i\ge0}[M_i,M_{i+1}]
$$
consisting of the segments $[M_i,M_{i+1}]$. The segment $[M_0,M_1]$ 
corresponds to the half line.
 \end{definition}
 
 \begin{remark}
The same polygonal line can be defined by using the conjugate partition.
Consider the points
$$
N_0=(0,+\infty),\,
N_1=(0,\lambda_1),\,
N_2=(1,\lambda'_1),\,
N_3=(1,\lambda'_2),\,
N_4=(2,\lambda'_2),\,
\dots\,.
$$
It is easy to check that 
$$
Y_{\lambda}=\bigcup_{j\ge0}[N_j,N_{j+1}]
$$
\end{remark} 

\begin{definition}  
Let $n,m$ be nonnegative integers. Define a bijective map 
$$
\vartheta_{n,m} : \Bbb R^2\longrightarrow \Bbb R^2,\quad \vartheta_{n,m}(x,y)=(n-x,m-y)
$$
\end{definition}

So for any bipartition $(\lambda,\mu)$  we can define a pair of polygonal  lines 
$$
Y_{\lambda,\mu}=(Y_{\lambda},\vartheta_{n,m}(Y_{\mu}))
$$
We can now reformulate the definition of the set $Cr(n,m)$ in the following way.
\begin{lemma} A bipartition $(\lambda,\mu)$ belongs to $Cr(n,m)$ if and only if the polygonal lines $Y_{\lambda}$ and $\vartheta_{n,m}(Y_{\mu})$ have at least one intersection point.
\end{lemma}

\begin{proof} Let $(\lambda,\mu)\in Cr(n,m)$ and $(p,s)$ the extremal  pair for them. 
Let us prove that the point $M=(r,p)$ is an intersection point. 
First let us prove that $M\in Y_{\lambda}$. Since $[M_{2p+1},M_{2p}]\subset Y_{\lambda}$ 
we only need to prove that $M\in [M_{2p+1},M_{2p}]$. 
By definition $\lambda_{p+1}\le r$. 
Suppose that $\lambda_p>r$. Then $\lambda_{p+1}\le\lambda_p\le r-1$. 
Therefore $\lambda\in H(p,r-1)$ and $\mu\in H(q,s)\subset H(q,s+1)$. 
But this contradicts the condition that the pair $(p,s)$ is the extremal one. 
So $\lambda_p\ge r$ and $M\in Y_{\lambda}$. 

In the same way we can prove that $\mu'_{q+1}\le s\le \mu'_q$. 
Therefore point $N=(s,q)\in [N_{2s+1},N_{2s}]\subset Y_{\mu}$ and 
$\vartheta_{n,m}(N)=M\in \vartheta_{n,m}(Y_{\mu})$.

Now let us prove the converse  statement. 
Let $M=(\tilde r,\tilde p)\in Y_{\lambda}\cap \vartheta(Y_{\mu})$. 
Then  it is easy  to check that
$$
Y_{\lambda}\in H(\tilde p,\tilde r),\quad Y_{\mu}\in  H(n-\tilde p,m-\tilde r)
$$
Therefore $(\lambda,\mu)\in Cr(n,m)$. The lemma is proved.
\end{proof}

\begin{corollary}  
Let $(\lambda,\mu)\in Cr(n,m)$ and $(p,s)$ be the corresponding extremal pair.
Then the point $M=(r,p)$ is the maximal  of the  intersection points of the lines 
$Y_{\lambda}$ and $\vartheta_{n,m}(Y_{\mu})$. 
\end{corollary}

Let us now define a map $\sigma : Cr(n,m)\longrightarrow X^{+}(n,m)$.
Let $(\lambda,\mu)\in Cr(n,m)$ and  $Y_{\lambda}, \vartheta(Y_{\mu})$ 
be the corresponding lines. Let $M$ be the maximal intersection point of these lines. 
We have  two disjoint unions
$$
Y_{\lambda}=Y^{L}_{\lambda}\cup\{M\}\cup Y^{R}_{\lambda},\quad \vartheta_{n,m}(Y_{\mu})=\vartheta_{n,m}(Y_{\mu})^{L}\cup\{M\}\cup \vartheta_{n,m}(Y_{\mu})^{R}
$$
where $Y_{\lambda}^{L}$ means the set of points  on the line $Y_{\lambda}$  
which are strictly less  than $M$, and   
$Y_{\lambda}^{R}$ means the set of points  on the line $Y_{\lambda}$   
which are strictly greater than $M$ . The same for the line $\vartheta_{n,m}(Y_{\mu})$ instead of $Y_{\lambda}$.
\begin{definition}
Let us define  the  following two sets  as disjoint unions 
$$
\Gamma_a=Y^{L}_{\lambda}\cup\{M\}\cup\vartheta_{n,m}( Y_{\mu})^{R},
\quad
\hat\Gamma_b=Y^{R}_{\lambda}\cup\{M\}\cup\vartheta_{n,m}( Y_{\mu})^{L}
$$
and $\pi(\lambda,\mu)=(a,b)$.
\end{definition}

\begin{proposition} We have the following equality
$$
\pi(\lambda,\mu)=\sigma(\lambda,\mu)
$$  
\end{proposition}
\begin{proof} It easily follows from the definition of the map $\sigma$.
\end{proof}

\section{Action on the generalised eigenspace}

In order to describe the action of the algebra  $\mathcal D_{n,m}$ 
on the subspaces $\Lambda^{\pm}_{n,m}(\chi)$ we will use the infinite 
dimensional version of the CMS  operators investigated  in  \cite{Special}. 
But first we will prove a more general abstract result. 

Let  $\varphi : A\rightarrow B$ be a surjective homomorphism of commutative algebras.
Let  $U,V$ be modules over the algebras $A,B$ respectively.
Let $\psi : U\rightarrow V$  be a linear surjective map such that $\psi(au)=\varphi(a)\psi(u)$. 
Suppose that both modules split into direct sums of
generalised finite dimensional  eigenspaces
$$
U=\bigoplus_{\chi\in S} U(\chi),\quad V=\bigoplus_{\tilde\chi \in T}
V(\tilde\chi)
$$
where $S\subset A^*,\, T\subset B^*$ are some sets of homomorphisms. 
  
The homomorphism $\varphi$ induces a map $\varphi^* : T\rightarrow S$ 
such that $\varphi^*(\tilde\chi)=\tilde\chi\circ\varphi$. 
Suppose we have a basis of $e_i,\,i\in I$  in $U$ and a basis of $f_j,j \in J$ in $V$ 
such that $e_i\in U(\chi_i)$ and  $f_j\in V(\tilde\chi_j)$ 
for some $\chi_i\in S$ and $\tilde\chi_j\in T$. 
Suppose there exists an injective map $\gamma: J\rightarrow I$ 
such that $\varphi^*(\tilde\chi_j)=\chi_{\gamma(j)}$.
  

\begin{lemma}\label{basis}  Suppose that all the previous conditions are fulfilled.
Then the set of vectors $ \{e_i\mid i\notin\operatorname{im}\gamma\}$ makes a basis of the space 
$
\ker \psi
$
and  if $\chi=\varphi^*(\tilde\chi)$ then the restriction of   $\psi$ to $U(\chi)$ is isomorphism onto $V(\tilde\chi)$.
\end{lemma}

\begin{proof} Let $\chi\in T$. If $\psi(U(\chi))\ne0$ then there exists $\tilde\chi$ such that 
$\tilde\chi=\chi\circ\varphi$. Therefore  $\chi\in\operatorname{im}\varphi^*$ 
and since $\psi$ is surjection, 
then the  restriction map  $\psi:U(\chi)\longrightarrow V(\tilde\chi) $ is  also a surjection. 
So we see that if $\chi\notin\operatorname{im}\varphi^*$ then $\psi(U(\chi))=0$. 
Therefore $e_i \in \ker\psi$  for $ i\notin\operatorname{im}\gamma$. 
 Let us prove now that if $\chi=\varphi^*(\tilde\chi)$ then $\psi: U(\chi)\rightarrow V(\tilde\chi)$ is isomorphism. We have already seen that this is a surjective map. Therefore it is enough to prove that $\dim U(\chi)=\dim V(\tilde\chi)$. Let us set 
$$
I_\chi=\{i\mid \chi_i=\chi\},\quad  J_{\tilde\chi}=\{j\mid \tilde\chi_j=\tilde\chi\}
$$
From the conditions of the present lemma it follows that 
$ \dim\, U(\chi)=|I_\chi|$ and  $\dim \,V(\tilde\chi)=|J_{\tilde\chi}|.$ So it is enough to prove that $|I_\chi|=|J_{\tilde\chi}|.$ It is easy to see that we only need to prove that $\gamma(J_{\tilde\chi})=I_\chi$. Let $j\in J_{\tilde\chi}$ then $\chi_{\gamma(j)}=\varphi^*(\tilde\chi_{j})=\varphi^*(\tilde\chi)=\chi$. Conversely let $i\in I_{\chi}$ and $i=\gamma(j)$, then $\varphi^*(\tilde\chi_j)=\chi_{\gamma(j)}=\chi_i=\chi.$ Since $\varphi^*$ is injective map we get $\tilde\chi_j=\tilde\chi$ and therefore $j\in J_{\tilde\chi}$.
\end{proof}

\begin{thm}\label {main1} 
Let $k$ be not an algebraic number and
$\dim \Lambda^{\pm}_{n,m}(\chi)=2^r$.
Then the  image of the algebra  $\mathcal D_{n,m}$ in the algebra  
$End( \Lambda^{\pm}_{n,m}(\chi))$  
is isomorphic to $\Bbb C[\varepsilon]^{\otimes r}$, 
where $\Bbb C[\varepsilon]$ is the algebra of dual numbers and  
the space  $\Lambda^{\pm}_{n,m}(\chi)$ 
is the regular representation with respect to this action.
\end{thm}

\begin{proof} 
We will use Lemma \ref{basis}. 
Let $A=\mathcal D$ be  the algebra of infinite integrals with $p_0=n+k^{-1}m$
(see \cite{Laurent}) and  $B=\mathcal D_{n,m}$.
Let $V=\Lambda^{\pm}$ be the algebra of Laurent symmetric functions 
and $U=\Lambda^{\pm}_{n,m}$. Let $I$ be  the set of 
all bipartitions and  $J=X^+(n,m)$. The basis in the space $V$ 
is the set of $Q_{\lambda,\mu}$ constructed in \cite{Special}.
The basis $f_{\chi}$ in the space $U$  the one was constructed in Proposition
\ref{max}.
Besides we have the following maps: 

$$
\psi : \Lambda^{\pm}\longrightarrow \Lambda^{\pm}_{n,m},\quad \psi(p_r)=x_1^s+\dots+x_n^s+\frac1k(x_{n+1}^s+\dots+x_{n+m}^s),\,s=\pm1,\pm2\dots
$$

and the  map $\varphi$ comes from the following commutative diagram \cite{Laurent}

$$
\begin{array}{ccc}
\Lambda^{\pm}&\stackrel{D}{\longrightarrow}&\Lambda^{\pm}\\ \downarrow
\lefteqn{\psi}& &\downarrow \lefteqn{\psi}\\
\Lambda^{\pm}_{n,m}&\stackrel{\varphi(D)}{\longrightarrow}& 
\Lambda^{\pm}_{n,m}\\
\end{array}
$$

where $D\in \mathcal D$.  The set $S$ is defined by 

$$
S=\{\theta_{\lambda,\mu}\mid \theta_{\lambda,\mu}(D)=\phi(D)(\lambda,\mu) \}
$$
where $\phi : \mathcal D\rightarrow \Lambda^*$ is the Harish-Chandra 
homomorphism from \cite {Laurent}, $\Lambda^*$ 
is the algebra of shifted symmetric functions and for $f\in \Lambda^*$ the value $f(\lambda,\mu)$ is defined in   \cite{Laurent}, page 79,  formula $(56)$.
The set $T$ is the set of all $\theta_{\chi}$ where $\chi\in X^{+}_{n,m}$.
The corresponding map $\gamma$  is the inverse map to $\pi$:
$$
\gamma: J\longrightarrow I,\quad  \gamma(\chi)=\pi^{-1}(\chi)
$$
So  in order to prove the Theorem we only  need to prove the equality 
$$
\theta_{\pi(\lambda,\mu)}\circ\varphi=\theta_{\lambda,\mu}
\quad\text{for}\quad (\lambda,\mu)\in Cr(n,m)
$$
But from the commutative diagram 
$$
\begin{array}{ccc}
\mathcal D&\stackrel{\phi}{\longrightarrow}&\Lambda^*\\ \downarrow
\lefteqn{\varphi}& &\downarrow \lefteqn{\varphi^{*}}\\
\mathcal D_{n,m}&\stackrel{\phi_{n,m}}{\longrightarrow}& 
\Lambda^*_{n,m}\\
\end{array}
$$
we see that enough to prove the equality  $f(\lambda,\mu)=\varphi^*(f)(\pi(\lambda,\mu))$  where $f\in\Lambda^*$. 

Let us identify Young diagram $\lambda$ with the part of the plane 
bounded by the  lines $x=0,\,y=0,\,Y_{\lambda}$. 
Let us recall that for any unit square such that all its vertices 
have integer coordinates we defined the function $c_k(\square)=x+ky$ 
where $(x,y)$ are the coordinates of the left lower vertex of the square. 

\begin{lemma}  
Let $f\in \Lambda^*$ and $(\lambda,\mu)\in Cr(n,m)$. Then
$$
f(\lambda,\mu)=\varphi^*(f)(\pi(\lambda,\mu))
$$
\end{lemma}
\begin{proof} It is enough to prove the above equality for  the shifted symmetric function 
corresponding to Bernoulli  polynomials
$$
b^{\infty}_r(x) =\sum_{i\ge 1}\left[B_r(x_i+k(i-1))-B_r(k(i-1))\right],\,\,\,r=1,2,\dots. 
$$
But in this case (see \cite{ Laurent} Lemma 5.8) we have 
$$
b^{\infty}_r(\lambda,\mu)=r\sum_{\square\in\lambda}c_k(\square)^{r-1}+r(-1)^r\sum_{\square\in\mu}(c_k(\square)+1+k-kn-m)^{r-1}
$$
$$
=r\sum_{\square\in\lambda}c_k(\square)^{r-1}-r\sum_{\square\in\vartheta_{n,m}(\mu)}c_k(\square)^{r-1}
$$
Let us denote by $\Delta$ the part of the plane bounded by the lines 
$x=0$, $\,y=n$, $Y_{\lambda}$, $\vartheta(Y_{\mu})$.  Since $\lambda$ and $\vartheta_{n,m}(\mu)$ 
both contain $\Delta$, we have 
$$
b^{\infty}_r(\lambda,\mu)=r\sum_{\square\in\lambda\setminus\Delta}c_k(\square)^{r-1}-r\sum_{\square\in\theta_{n,m}(\mu)\setminus\Delta}c_k(\square)^{r-1}=
$$
$$
r\sum_{\square\in D_a^+\cup \hat D_{b}^+}c_k(\square)^{r-1}-r\sum_{\square\in D_a^-\cup \hat D^-_b}c_k(\square)^{r-1}=b^{(n,m)}_r(\pi(\lambda,\mu))
$$
We only need to prove that $\varphi^*(b^{\infty}_r)=b^{(n,m)}_r$.  In general Jack-Laurent symmetric functions $P_{\lambda,\mu}$ are eigenfunctions of the algebra $\mathcal D$.  In our case $p_0=n+k^{-1}m$  so not all  the Jack-Laurent symmetric functions are well defined. But if we take 
 $\lambda=(\lambda,\emptyset)$ a bipartition with the empty second part then
the corresponding Jack-Laurent symmetric function $P_{\lambda,\emptyset}$ does not depend on $p_0$ therefore  it is well define and  it is simply the Jack symmetric function $P_{\lambda}$ corresponding to the partition $\lambda$. Therefore by definition we have
$$
\phi_{n.m}(D)(\chi)\psi(P_{\lambda})=\psi(D)(\lambda,\emptyset)\psi(P_{\lambda}),\,\,D\in\mathcal D.
$$
For $\lambda$ with $\lambda_{n}\ge m$ 
the highest weight of the polynomial  $\psi(P_{\lambda})$ is 
$$
\pi(\lambda,0)= (\lambda_1,\dots,\lambda_n\mid \lambda'_1-n,\dots,\lambda'_m-n)
$$
(see \cite{Gendis}). So if we  take $D$ such that $\phi(D)=b_r^{\infty}$ then  by definition 
$$
\varphi^*(b^{\infty}_r)((\lambda_1,\dots,\lambda_n | \lambda'_1-n,\dots,\lambda'_m-n))=b^{(n,m)}_r((\lambda_1,\dots,\lambda_n | \lambda'_1-n,\dots,\lambda'_m-n))
$$
for all $\lambda$ with $\lambda_{n}\ge m$. Therefore $\varphi^*(b^{\infty}_r)=b^{(n,m)}_r$.
\end{proof}
Thus we have proved the Lemma and the Theorem.
\end{proof}
 The parameter $p_0$ plays a special  role in the theory of Jack-Laurent symmetric function. This parameter can be considered as a natural generalisation of  the dimension.  If $p_0=n+k^{-1}m$ then  as it was proved in the Theorem \ref{main1} the infinite dimensional  integrable system has a finite dimensional quotient.   The next corollary of the Theorem \ref{main1} shows that  the kernel of the homomorphism $\psi$  can be described in terms of the Jack-Laurent polynomials  which are well defined at $p_0=n+k^{-1}m$ and eigenfunctions of the algebra $\mathcal D_{n,m} $ can be constructed in terms of Jack-Laurent symmetric functions as well.
 
\begin{corollary} 

$1)$ Let $k$ be not a rational number  
and $(\lambda,\mu)\notin Cr(n,m)$. Then the Jack-Laurent 
polynomial $P_{\lambda,\mu}$ is well defined at $p_0=n+k^{-1}m$ 
and  the linear span of all such  $P_{\lambda,\mu}$  is the kernel of the homomorphism $\psi$.

$2)$ Let $\chi\in X_{reg}^+(n,m)$. Then generalised eigenspace $\Lambda^{\pm}_{n,m}(\chi)$ contains the only eigenfunction $J_{\chi}$  of the algebra $\mathcal D_{n,m}$. The Jack-Laurent symmetric function $P_{\pi^{-1}(\chi)}$ is well defined at $p_0=n+k^{-1}m$ and $J_{\chi}=d\psi(P_{\pi^{-1}(\chi)})$, where $d \in \Bbb C$ is a  constant.
 \end{corollary}
\begin{proof}
This corollary follows directly from  
Lemma \ref{basis},  from the proof of the  Theorem \ref{main1} and the Theorem 3.6 from \cite{Special}.
\end{proof}

\section{Concluding remarks}

We have shown that the spectral  decomposition of the algebra of
quasi-invariant Laurent polynomials under the action of the algebra of
deformed integrals has a nice description, despite the fact that 
this decomposition is not multiplicity free.  

One of the main steps in the proof is using an
infinite dimensional version of the CMS problem
to describe  the algebra of endomorphisms of a generalised eigenspace. 
It would be good to prove this fact directly, without  
appealing to the infinite dimensional version. The Moser matrix as a 
linear operator on the quasi-homomorphisms seems to be a key ingredient 
in such a proof. 

It also looks like  that the group $G$  
generated by reflections with respect to the deformed inner product 
should play an important role in the whole theory of deformed CMS systems 
and of the  corresponding groupoids.   

In the paper  \cite{ER}  a natural  generalisations of the algebra $\Lambda_{n,m}^{+} $ were investigated. It would be interesting to investigate the same sort of generalisations  of the  algebras $\Lambda_{n,m}^{\pm}$ as well.
\section{Acknowledgements}
This work  has been funded  by  Russian Ministry of Education  and  Science  
(grant 1.492.2016/1.4) and  
partially by the  Russian  Academic Excellence Project '5-100'.

\end{document}